\documentclass[10pt,conference,letterpaper]{IEEEtran}

\usepackage{tikz}
\usepackage{pgfplots}
\usepackage[]{subfig}
\usepackage{amsmath}
\usepackage{algorithm}
\usepackage[noend]{algpseudocode}
\usepackage{multirow}
\usepackage{amsthm}

\definecolor{darkgreen}{rgb}{0.125,0.5,0.169}

\usepackage[show]{chato-notes}

\newtheoremstyle{style}
  {\topsep} 
  {\topsep} 
  {} 
  {} 
  {\bfseries} 
  {.} 
  {.5em} 
  {} 

\newtheorem{problem}{Problem}

\newtheorem{definition}{Definition}
\newtheorem{proposition}{Proposition}
\newtheorem{lemma}{Lemma}
\newtheorem{theorem}{Theorem}

\newcommand{\squishlist}{\begin{list}{$\bullet$}
  { \setlength{\itemsep}{0pt}
     \setlength{\parsep}{3pt}
     \setlength{\topsep}{3pt}
     \setlength{\partopsep}{0pt}
     \setlength{\leftmargin}{1.5em}
     \setlength{\labelwidth}{1em}
     \setlength{\labelsep}{0.5em} } }
\newcommand{\squishend}{
  \end{list}  }

\newcommand{\problemqbff}{\ensuremath{\textsc{Qr-Bff}}}
\newcommand{\problembff}{\ensuremath{\textsc{Bff}}}
\newcommand{\problemtbff}{\ensuremath{\textsc{O$^2$Bff}}}
\newcommand{\problemtbffmm}{\ensuremath{\textsc{O$^2$Bff-mm}}}
\newcommand{\problemtbffma}{\ensuremath{\textsc{O$^2$Bff-ma}}}
\newcommand{\problemtbffam}{\ensuremath{\textsc{O$^2$Bff-am}}}
\newcommand{\problemtbffaa}{\ensuremath{\textsc{O$^2$Bff-aa}}}

\newcommand{\algotbff}{\ensuremath{\textsc{FindO$^2$Bff}}}

\newcommand{\algobff}{\ensuremath{\textsc{FindBff}}}
\newcommand{\algobffmm}{\ensuremath{\textsc{FindBff}_\textsc{M}}}
\newcommand{\algoQbffmm}{\ensuremath{\textsc{Qr-FindBff}_\textsc{M}}}

\newcommand{\algobffaa}{\ensuremath{\textsc{FindBff}_\textsc{A}}}
\newcommand{\algoQbffaa}{\ensuremath{\textsc{Qr-FindBff}_\textsc{A}}}

\newcommand{\algobffgreedy}{\ensuremath{\textsc{FindBff}_\textsc{G}}}

\newcommand{\problemaa}{\ensuremath{\textsc{Bff-aa}}}
\newcommand{\problemmm}{\ensuremath{\textsc{Bff-mm}}}
\newcommand{\problemQmm}{\ensuremath{\textsc{Qr-Bff-mm}}}
\newcommand{\problemQaa}{\ensuremath{\textsc{Qr-Bff-aa}}}

\newcommand{\problemam}{\ensuremath{\textsc{Bff-am}}}
\newcommand{\problemma}{\ensuremath{\textsc{Bff-ma}}}
\newcommand{\degree}{\ensuremath{\textit{degree}}}

\newcommand{\itg}{\ensuremath{\textsc{ITR}}}
\newcommand{\itk}{\ensuremath{\textsc{ITR}_{K}}}
\newcommand{\itr}{\ensuremath{\textsc{ITR}_{R}}}
\newcommand{\itc}{\ensuremath{\textsc{ITR}_{C}}}

\newcommand{\ind}{\ensuremath{\textsc{INC}_{D}}}
\newcommand{\ino}{\ensuremath{\textsc{INC}_{O}}}

\newcommand{\saa}{\ensuremath{\textit{aa}}}
\newcommand{\sam}{\ensuremath{\textit{am}}}
\newcommand{\sma}{\ensuremath{\textit{ma}}}
\newcommand{\smm}{\ensuremath{\textit{mm}}}

\newcommand{\spara}[1]{\smallskip\noindent{\bf{#1}}}

\newcommand{\etal}{\emph{et al.}}

\newcommand{\db}{\textit{DBLP}}
\newcommand{\dbten}{\textit{DBLP$_{10}$}}

\newcommand{\as}{\textit{AS}}
\newcommand{\cai}{\textit{Caida}}
\newcommand{\orf}{\textit{Oregon$_1$}}
\newcommand{\ort}{\textit{Oregon$_2$}}

\newcommand{\twitter}{\textit{Twitter}}

\newcommand{\calG}{\ensuremath{{\mathcal G}}}
\newcommand{\calC}{\ensuremath{{\mathcal C}}}

\begin{document}

\title{\LARGE \textbf{Best Friends Forever (BFF): Finding Lasting Dense Subgraphs}}

\author{
{Konstantinos Semertzidis{\small $~^{1}$}, Evaggelia Pitoura{\small $~^{1}$}, Evimaria Terzi{\small $~^{2}$}, Panayiotis Tsaparas{\small $~^{1}$} }%
\vspace{1.6mm}\\
\fontsize{10}{10}\selectfont\itshape
$^{1}$\,Dept. of Computer Science and Engineering, University of Ioannina, Greece\\
$^{2}$\,Dept. of Computer Science, Boston University, USA\\
\fontsize{9}{9}\selectfont\ttfamily\upshape

$^{1}$\,\{ksemer,pitoura,tsap\}@cs.uoi.gr\\
$^{2}$\,evimaria@cs.bu.edu%
}

\maketitle

\begin{abstract}
Graphs form a natural model for relationships and interactions between entities, for example, between people in social and cooperation networks, servers  in computer networks,
or tags and words in documents and tweets. But, which of these relationships or interactions are the most lasting ones?
	In this paper, we study the following problem: given a set of graph snapshots, which may correspond to the state of an evolving graph at different time instances, identify the set of nodes that are the most densely connected in all snapshots. We call this problem the {\it Best Friends For Ever}  ({\problembff}) problem.
	We provide definitions for density over multiple graph snapshots, that capture different semantics of connectedness over time, and
	we study the corresponding variants of the {\problembff} problem.
	We then look at the {\it On-Off} {\it {\problembff}}  ({\problemtbff}) problem that relaxes the requirement of nodes being connected in all snapshots, and asks for the densest set of nodes in at least $k$ of a given set of graph snapshots. We show that this problem is NP-complete for all definitions of density, and we propose a set of efficient algorithms. Finally, we present experiments with synthetic and real datasets that show both the efficiency of our algorithms and the usefulness of the {\problembff} and the {\problemtbff} problems.
\end{abstract}

\section{Introduction}
Graphs offer a natural model for capturing the interactions
and relationships among entities.
Oftentimes, multiple snapshots of a graph are available; for example, these snapshots may
correspond to the states of a dynamic graph at different time instances, or the states of a complex system at different conditions.
We call such sets of graph snapshots, a {\it graph history}.
Analysis of the graph history finds a large spectrum of applications, ranging from social-network marketing,
to virus propagation and digital forensics.
A  central question in this context is: {\it which interactions, or relationships in a graph history are the most lasting ones?}
In this paper, we formalize this question and we design algorithms that effectively identify such relationships.

In particular, given a graph history, we  introduce the problem of efficiently finding the set of nodes, that remains the most tightly connected through history.
We call this problem  the {\it Best Friends For Ever}  ({\problembff}) problem.
We formulate the {\problembff} problem as the problem of locating the set of nodes that have the maximum \emph{aggregate density} in the graph history.
We provide different definitions for the aggregate density that capture different notions of connectedness over time, and  result
in four variants of the {\problembff} problem.

We then extend the {\problembff} problem to capture the cases where subsets of nodes are densely connected for only a subset of the snapshots. Consider for example, a set of
collaborators that work intensely together for some years and then they drift apart, or, a set of friends in a social network that stop interacting for a few snapshots and then, they reconnect with each other.
To identify such subsets of nodes, we define the \emph{On-Off} {\problembff} problem, or {\problemtbff} for short. In the {\problemtbff}
problem, we ask for a set of nodes and a set of $k$ snapshots such that the aggregate density of the nodes over these snapshots is maximized.

Identifying  {\problembff} nodes finds many applications.
For example,
in collaboration and social networks such nodes correspond to well-acquainted individuals, and they can be chosen to form teams, or organize
successful professional or social events.
In a protein-interaction network, we can locate
protein complexes that are densely interacting
 at different states, thus indicating a possible underlying regulatory mechanism.
In a network where nodes are words or tags and edges correspond to their co-occurrences in documents or tweets published in a specific period of time,
identifying {\problembff} nodes may serve as a first step in topic identification, tag recommendation and other types of analysis.
In a computer network, locating servers that communicate heavily over time may be useful in identifying potential attacks, or bottlenecks.

The problem of identifying a dense subgraph in a static (i.e., single-snapshot) graph has received a lot of attention
(e.g., ~\cite{DBLP:conf/approx/Charikar00,goldberg1984finding,DBLP:conf/icalp/KhullerS09}).
There has been also work on finding dense subgraphs in dynamic graphs
(e.g., ~\cite{DBLP:conf/www/EpastoLS15}). However, in this line of work, the goal is to efficiently locate the densest subgraph in the current graph snapshot, whereas we are interested in locating subgraphs that remain dense in the whole graph history.
To the best of our knowledge, we are the first to systematically introduce and study density in a graph history, and define the {\problembff} and {\problemtbff} problems.
The most related work to ours is \cite{jethava15finding} where the authors
study just one of the four variants of the  {\problembff} problem in the context of graph databases. We compare the performance of our algorithms for this variant with the algorithm proposed in \cite{jethava15finding} experimentally.

We study the complexity of the different variants of the {\problembff} and {\problemtbff} problems. Two of the {\problembff} variants can be solved optimally, while the {\problemtbff} is NP-hard. We propose a generic algorithmic framework for solving our problems, that works in linear time.
Experimental results with real and synthetic datasets show the efficiency and effectiveness of
our algorithms in discovering lasting dense subgraphs.
Two case studies on bibliographic collaboration networks, and hashtag co-occurrence networks in {\twitter}  validate our approach.

To summarize,  the main contributions of this work are the following:
\squishlist
\item
We introduce the novel {\problembff} and {\problemtbff} problems of identifying  a subset of nodes that define dense subgraphs in a graph history.
To this end, we  extend the notion of density for graph histories, and provide definitions that capture different semantics of density over time leading to  four variants of our problems.

\item
We study the   complexity  of the variants of the {\problembff}  and {\problemtbff} problems and propose appropriate  algorithms. 
We prove the optimality, or the approximation factor of our algorithms whenever possible. 

\item We extend our definitions and algorithms to identify the {\problembff}s of an input set of query nodes.

\item We perform experiments with both real and synthetic datasets and demonstrate that our problem definitions are meaningful,
and that our algorithms work well in identifying dense subgraphs in practice.
\squishend

\vspace*{-0.05in}
\spara{Roadmap:}
In Section~\ref{sec:densities}, we provide definitions of aggregate density.
We introduce the {\problembff} problem  and  its algorithms in Section~\ref{sec:bff}, and the
{\problemtbff} problem and its algorithms in Section~\ref{sec:tbff}, while in Section~\ref{sec:generalized} we study extensions to the original problem.
Our experimental evaluation is presented in Section~\ref{sec:experiments} and comparison with related work in Section~\ref{sec:rw}. Section~\ref{sec:conclusions} concludes the paper.

\section{Aggregate density}\label{sec:densities}
%
We assume that we are given as input multiple graph snapshots over the same set of nodes.
Snapshots may be ordered, for example, when the snapshots correspond to the states of a dynamic graph.
We may also have an unordered collection of graphs, for example, when the snapshots correspond to graphs collected as a result of some scientific experiments. We refer to such graph collections as a \emph{graph history}.

\begin{definition} [{\sc Graph History}]
A  graph history  ${\calG}$ = $\{G_1$, $G_2$, $\dots,$ $G_{\tau}\}$ is a collection  of $\tau$ graph snapshots,
	where each snapshot $G_t = (V, E_t)$, $t$ $\in$ $[1, \tau]$, is defined over the same set of nodes $V$.
\end{definition}
An example of a  graph history with four snapshots is shown in Figure \ref{fig:example}. Note that our definition is applicable to graph snapshots with different set of nodes by considering $V$ as their union.

\begin{figure*}
	\centering
	\subfloat[$G_1$]{\includegraphics[width= 0.21\textwidth]{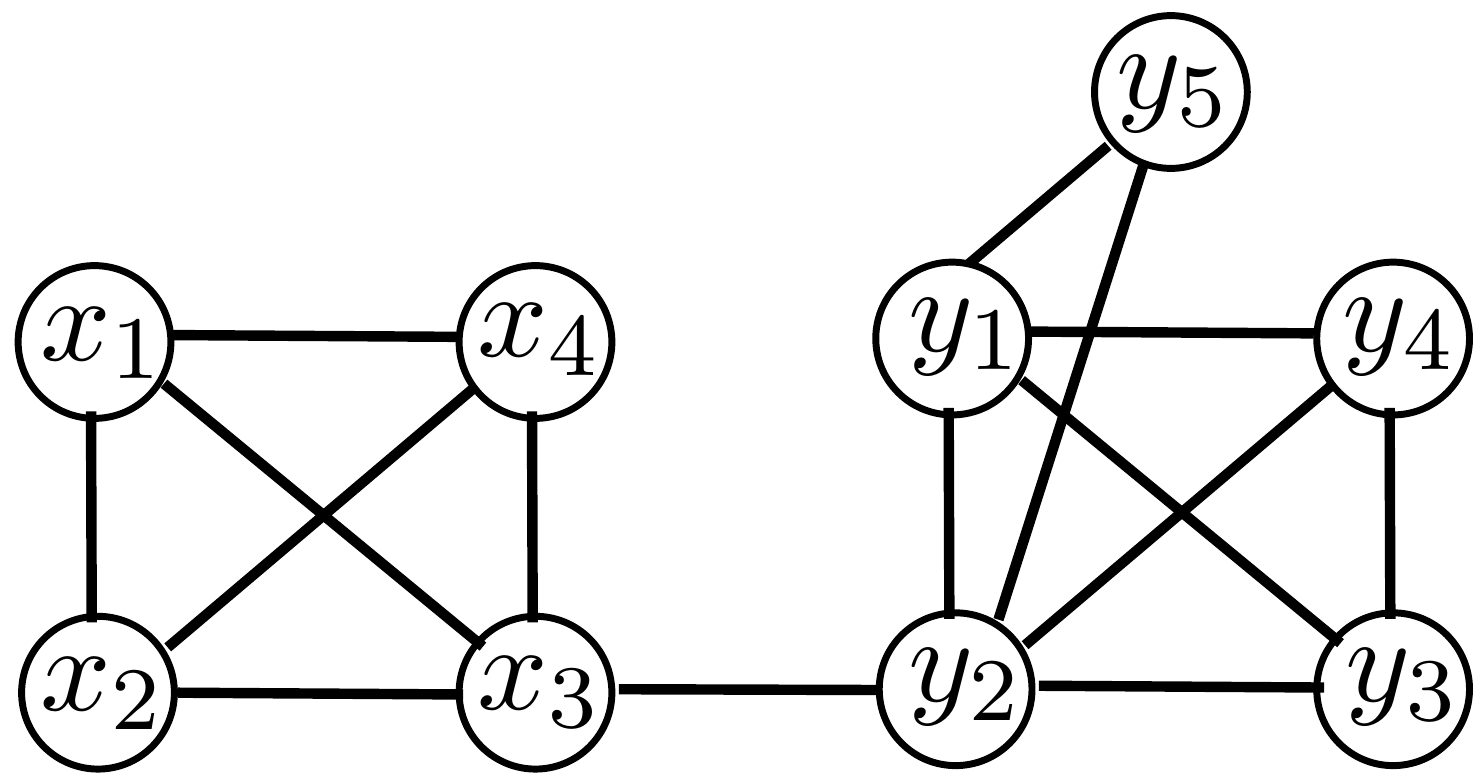}\label{fig:snapshot1}}
	\hspace{\fill}
	\subfloat[$G_2$]{\includegraphics[width= 0.21\textwidth]{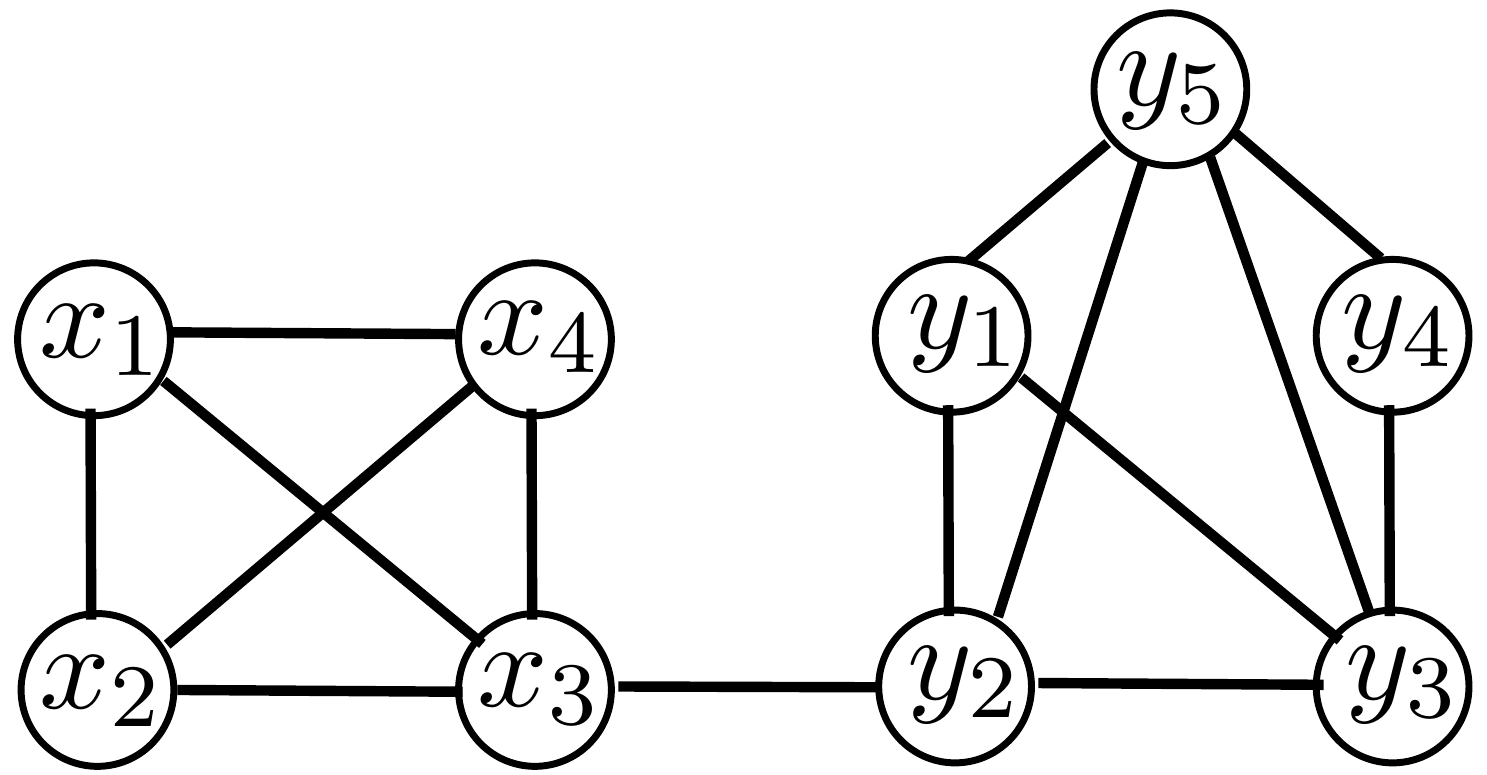}\label{fig:snapshot2}}
	\hspace{\fill}
	\subfloat[$G_3$]{\includegraphics[width= 0.21\textwidth]{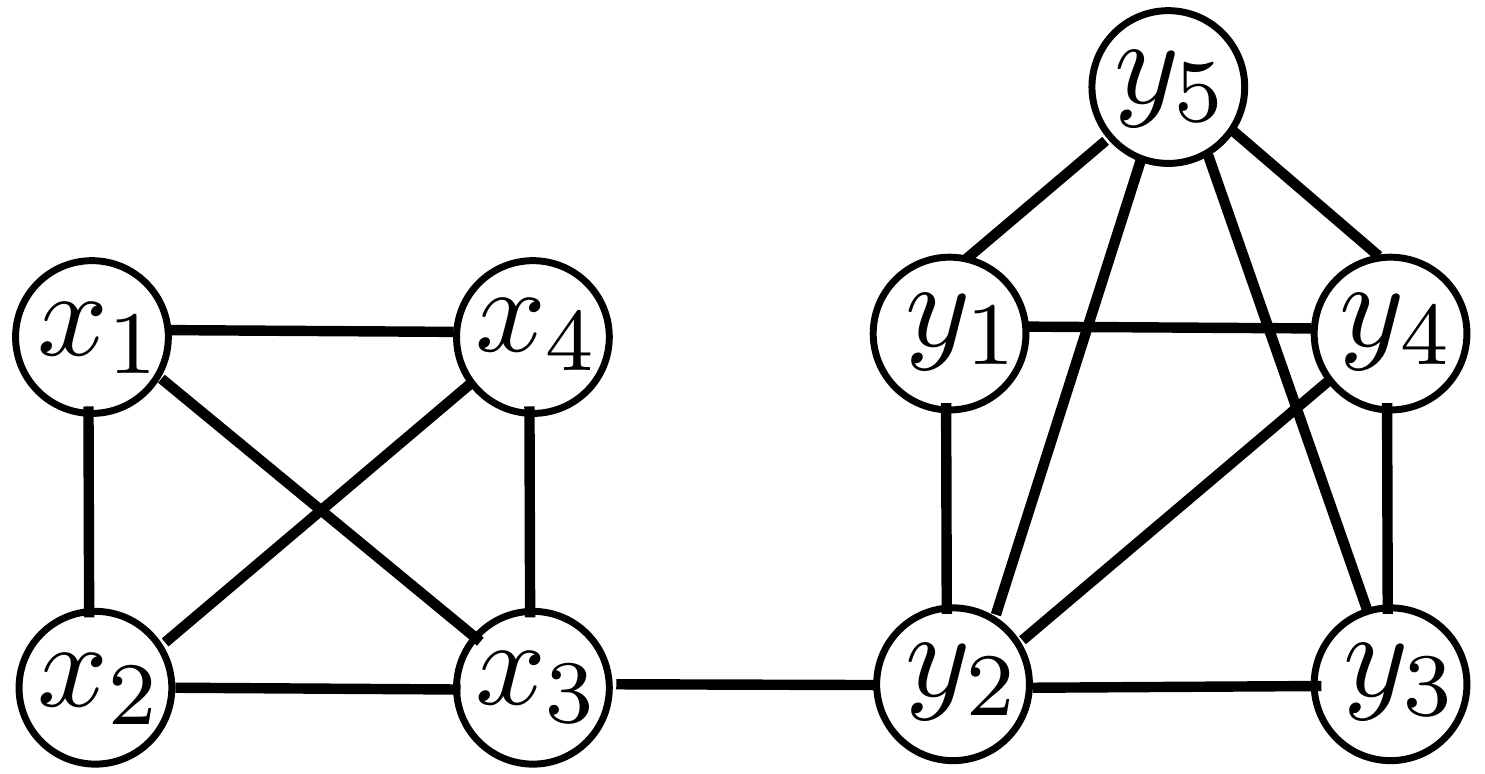} \label{fig:snapshot3}}
	\hspace{\fill}
	\subfloat[$G_4$]{\includegraphics[width= 0.21\textwidth]{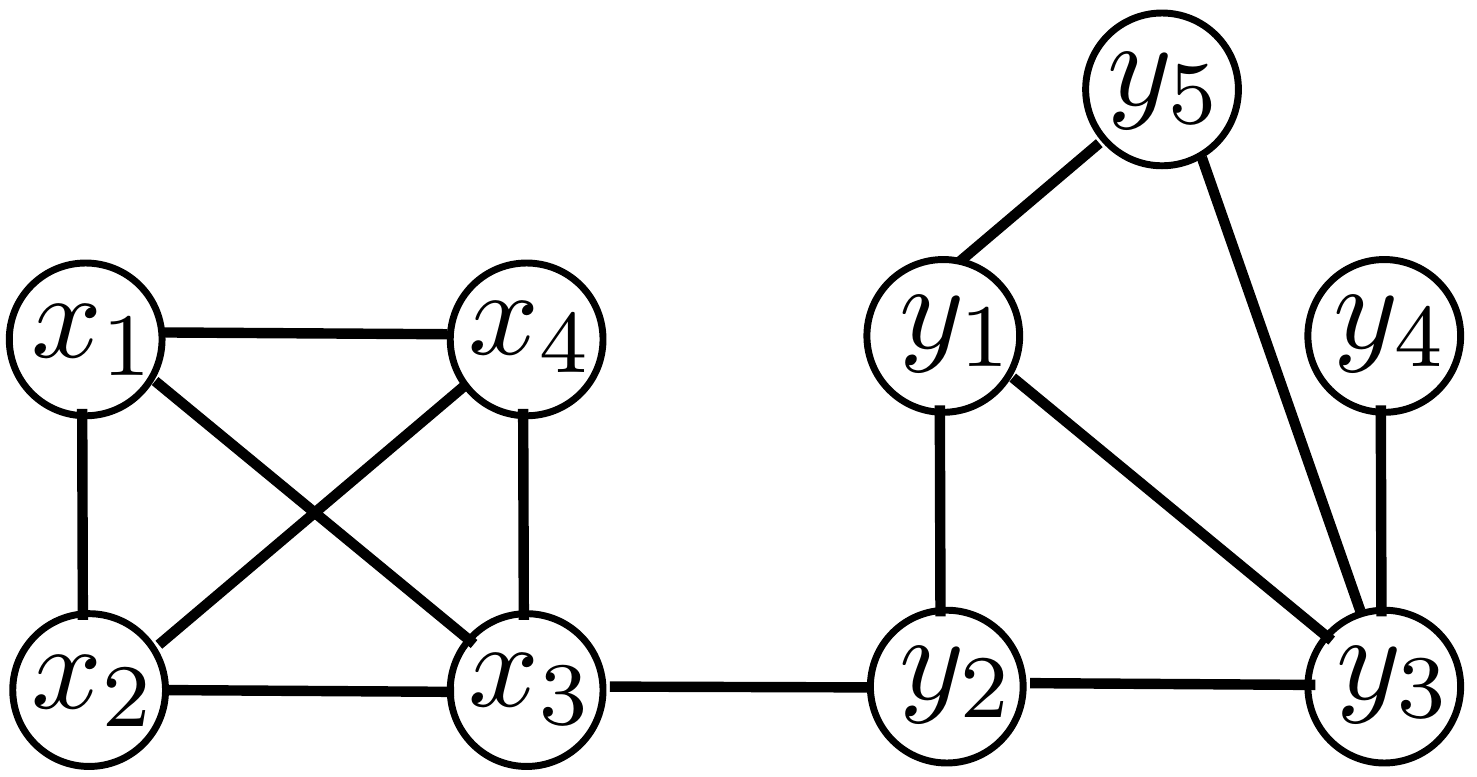} \label{fig:snapshot4}}
	\caption{\label{fig:example} A graph history $\calG=\{G_1,\ldots,G_4\}$ consisting of four snapshots.}
\end{figure*}

We will now define the notion of density for a graph history.
We start by reviewing two basic definitions of graph density for a single graph snapshot.
Given an undirected graph $G = (V, E)$ and a node  $u$ in $V$, we use $\degree(u,G)$ to denote the degree of $u$ in $G$.

The \emph{average density}, $d_a(G)$, of the graph $G$ is the average degree of the nodes in $V$:
\[
d_a(G)=\frac{1}{|V|}\sum_{u\in V}\degree(u,G) =  \frac{2|E|}{|V|}
\]
while the \emph{minimum density}, $d_m(G)$, of the graph is the minimum degree of any node in $V$:
\begin{align*}
d_m(G) = \min_{u \in V}degree(u, G).
\end{align*}

Intuitively, for a given graph, $d_m$  is defined by a single node, the one with the minimum degree, while
$d_a$ accounts for the degrees and thus the connectivity of all nodes.
For example, in Figure~\eqref{fig:snapshot1}, $d_m(G_1)=2$, while
$d_a(G_1)=10/3$.
Clearly, $d_m$ is a lower bound for $d_a$.
From now on, when the subscript of $d$ is ignored, density can be either $d_a$ or $d_m$.

We  also  define the density of a subset of nodes $S \subseteq V$ in the graph $G = (V,E)$. To this end, we use the induced subgraph $G[S] = (S, E(S))$ in $G$, where
$E(S) = \{(u,v) \in E: u \in S, \, v \in S\}$. 
We define the density $d(S,G)$ of $S$ in $G$ as $d(G[S])$. For example, again for snapshot $G_1$ in Figure~\ref{fig:example},
for $S_x=\{x_1,x_2,x_3,x_4\}$, $d_m(S_x,G_1) = d_a(S_x,G_1)=3$, while
for  $S_y=\{y_1,y_2,y_3,y_4,y_5\}$,  $d_m(S_y,G_1) =2$ and  $d_a(S_y,G_1)=16/5$.
Between $S_x$ and $S_y$, $S_x$ has the highest minimum density, whereas $S_y$ the highest average density.

We now define the density of a set of nodes $S$ on a graph history.
For this,
we need a way to aggregate the density of a set of nodes over multiple graph snapshots.

\spara{Aggregating density sequences:}
Given a graph history $\calG=\{G_1,\ldots,$ $G_\tau\}$, we will use $d(S,\calG)$ $= \{d(S,G_1),\dots,d(S,G_\tau)\}$ to denote the sequence of density values for the graph induced by the set $S$ in the graph snapshots.
We consider two definitions for an \emph{aggregation function} $g(d(S,\calG))$ that aggregates the densities over snapshots:
the first, $g_m$, computes the minimum density over all snapshots:
\[
g_m(d(S,\calG)) = \min_{G_t\in\calG} d(S,G_t),
\]
while the second, $g_a$, computes the average density over all snapshots:
\[
g_a(d(S,\calG)) = \frac {1}{\left|\calG\right|} \sum_{G_t\in\calG} d(S,G_t).
\]
Intuitively, the minimum aggregation function requires high density in each and every snapshot, while the average aggregation function looks
at the snapshots as a whole.
Again, we use $g$ to collectively refer to $g_m$ or $g_a$.
%
We can now define the \emph{aggregate density} $f$.

\vspace*{-0.005in}
\begin{definition} [{\sc Aggregate Density}]
	Given a graph history $\calG=\{G_1,\ldots ,G_\tau\}$ defined over a set of nodes $V$
	and $S\subseteq V$, we define the \emph{aggregate density} $f(S,\calG)$ to be
	$f(S,\calG)=g(d(S,\calG))$.
	Depending on the choice of the density function $d$ and the aggregation function $g$, we have the following four
	versions of  $f$:
(a) $	f_\smm(S,\calG)  =  g_m(d_m(S,\calG))$,
(b) $f_\sma(S,\calG)  =  g_m(d_a(S,\calG))$,
(c) $f_\sam(S,\calG)  =  g_a(d_m(S,\calG))$, and (d)
$	f_\saa(S,\calG)  =  g_a(d_a(S,\calG))$.
\end{definition}

\vspace*{-0.009in}
Each density definition associates different semantics with the meaning of density among nodes in a graph history.
Large values of $f_\smm(S,\calG)$
correspond to groups of nodes $S$
where each member of the group is connected with a large
number of other members of the group at each snapshot.
A node ceases to be considered a member of the group, if it loses touch with the other members
even in a single snapshot. 

Large values of $f_\sma(S,\calG)$ are achieved
for groups with high average density at each snapshot $G\in\calG$.
As opposed to $f_\smm(S,\calG)$, where the requirement is placed at each member of the group, large values of $f_\sma(S,\calG)$ are indicative
that the group $S$ has persistently high density as a whole.

The $f_\saa(S,\calG)$ metric takes large values when
the group $S$ has many connections on average; thus, $f_\saa$ is more ``loose"
both in terms of consistency over time and in terms of requirements at the individual group member level.

Lastly, $f_\sam(S,\calG)$ takes the average of the minimum
degree node at each snapshot, thus is less sensitive to the density of $S$ at a single instance.

For example, in the graph history $\calG$ in Figure~\ref{fig:example},
all aggregate densities for $S_x$ are equal to 3.
However, for  $S_y$
$f_{\saa}(S_y,\calG)$ = $31/10$,  while $f_{\sma}(S_y,\calG)$ = $12/5$.
That is, while $f_{\saa}(S_y,\calG) > f_{\saa}(S_x,\calG)$, $f_{\sma}(S_y,\calG) < f_{\sma}(S_x,\calG)$ due to the last instance.
Note also that $f_{\smm}(S_y,\calG)$ = $1$ and that this value is determined by just one node in just one snapshot, i.e., node $y_4$ in the last snapshot, while $f_{\sam}(S_y,\calG)$ = $2$.

\spara{The average graph:}
Finally, let us define the
\emph{average graph} of a graph history $\calG$ which is an edge-weighted graph where the weight of an edge is
equal to the fraction of snapshots in $\calG$ where the edge appears.

\begin{definition} [{\sc Average Graph}]
Given a graph history $\calG=\{G_1,\dots,G_\tau\}$ on a set of nodes $V$,
the average graph $\widehat{H}_\calG=(V,\widehat{E},\widehat{w})$ is a \emph{weighted, undirected} graph on the set of nodes
$V$, where $\widehat{E}$ = $V\times V$, and for each $(u,v)\in\widehat{E}$,
$\widehat{w}(u,v) = \frac{\left|G_t=(V,E_t)\in\calG\mid (u,v)\in E_t\right|}{\left|\calG\right|}.
$
\end{definition}

As usual, the degree of a node $u$
in a weighted graph is defined as: $\degree(u,\widehat{H}_\calG) =\sum_{(u,v)\in \widehat{E}}\widehat{w}(u,v)$.
The average graph performs aggregation on a per-node basis, in that, the degree of each node $u$ in $\widehat{H}_\calG$ is the average degree of $u$ in time.

With the average graph. we lose information regarding density at individual snapshots.
With some algebraic manipulation, we can prove the following lemma that shows a connection between the average graph and the $f_\saa$ density function:

\vspace*{-0.05in}
\begin{lemma} Let $\calG=\{G_1,\ldots ,G_\tau\}$ be a graph history over a set of nodes $V$ and $S$ a subset of nodes in $V$, it holds: $f_\saa(S,\calG)$ $=$  $d_a\left(\widehat{H}_\calG\left[S\right]\right)$.
\label{lemma:equivalence}
\end{lemma}
\section{The {\Large{\problembff}} Problem}
\label{sec:bff}

In this section, we introduce the {{\problembff}} problem, we study its hardness and propose appropriate algorithms.

\subsection{Problem definition}

Given the snapshots of a graph history ${\calG}$, our goal is to locate the Best Friends For Ever ({\problembff}),
that is, to identify a subset of nodes of $V$ such that
these nodes remain densely connected
with each other in \emph{all} snapshots of {\calG}. Formally:

\begin{problem}[The Best Friends Forever ({\problembff}) Problem]
\label{pr:bff}
Given a graph history ${\calG}$
and an aggregate density function $f$, find a subset of nodes $S\subseteq V$,
such that $f(S,\calG)$ is maximized.
\end{problem}

By considering the four choices for the aggregate density  function $f$, we have four variants of the {\problembff} problem. Specifically,
$f_\smm$, $f_\sma$, $f_\sam$ and $f_\saa$ give rise to problems:
{\problemmm}, {\problemma}, {\problemam} and {\problemaa} respectively.

\subsection{{\large {\problembff}} algorithms}\label{sec:bffalgos}


We now introduce a generic algorithm for the {\problembff} problem. The algorithm (shown in Algorithm~\ref{algo:bff})
is a ``greedy-like" algorithm inspired by a popular algorithm for the densest subgraph problem on a static graph~\cite{DBLP:conf/swat/AsahiroITT96,DBLP:conf/approx/Charikar00}.
We use $\calG[S]$  $= \{G_1[S],\ldots,G_\tau[S]\}$ to denote the sequence of
the induced subgraphs of the set of nodes $S$.
The algorithm
starts with a set of nodes $S_0$ consisting of all nodes $V$, and then it performs $n-1$ steps, where at each step $i$ it produces a set $S_i$ by removing one of the nodes in the set $S_{i-1}$. It then finds the set $S_i$ with the maximum aggregate density $f(S,\calG)$.

\begin{algorithm}
	\caption{\small The {\algobff}  algorithm.}
	\label{algo:bff}
	\begin{algorithmic}[1]
	\small
		\Statex{{\bf Input:} Graph history $\calG = \{G_1,\dots,G_\tau\}$; aggregate density function $f$}
		\Statex{{\bf Output:} A subset of nodes $S$}
		\vspace{0.1cm}
		\hrule
		\vspace{0.1cm}
        \State $S_0 = V$
		\For{$i = 1,\dots, n-1$}
            \State $v_i=\displaystyle\arg\min_{v\in S_{i-1}}\textit{score}\left(v,\calG\left[S_{i-1}\right]\right)$
            \label{line:select}
			\State $S_i = S_{i-1} \setminus \{v_i\}$
		\EndFor
		\State \textbf{return} $\arg\displaystyle\max_{i=0\ldots n-1}f(S_i,\calG)$\label{line:evaluation}
        \label{line:target}
	\end{algorithmic}
\end{algorithm}

The $\algobff$ algorithm forms the basis for the algorithms we propose for the four variants of the {\problembff} problem.
Interestingly, by defining appropriate  scoring functions,  $\textit{score}\left(v,\calG\left[S\right]\right)$,
 (used in line~\ref{line:select} to select which node to remove), we can get efficient algorithms for each of the variants.

\begin{algorithm}[ht!]
	\caption{\small The score$_m$ algorithm.}
	\label{algo:score_m}
	\begin{algorithmic}[1]
		\small
		\Statex{{\bf Input:} Graph history $\calG = \{G_1,\dots,G_\tau\}$}
		\Statex{{\bf Output:} Node with the minimum score$_m$}
		\vspace{0.1cm}
		\hrule
		\vspace{0.1cm}
        \State ${\cal L}_{t}$[$d$] $\leftarrow$ list of nodes with degree $d$ in $G_t$\\

        \Procedure{ScoreAndUpdate}{$ $}
        \For{$t = 1,\dots,\tau$}
        	\State $dmin_t$ $\leftarrow$ smallest $d$ s.t. ${\cal L}_t[d] \neq \emptyset$
		\EndFor
		\State $score_m$ = $\displaystyle\min_{t = 1,\dots,\tau}{dmin_t}$
		\State $t'$ = $\displaystyle\arg\min_{t = 1,\dots,\tau}{dmin_t}$
		\State $u$ = ${\cal L}_{t'}[score_m]$.get()
		\For{\textbf{each} $G_t \in \calG$}
			\State ${\cal L}_{t}$[$score_m$].remove($u$)
			\For{\textbf{each} $(u, v) \in E_t$}
				\State ${\cal L}_{t}$[degree($v$, $G_t$)].remove($v$)
				\State $E_t = E_t - (u, v)$ \textit{\,\,\,\,// update $\displaystyle degree_{v \in V}(v, G_t)$}
				\State ${\cal L}_{t}$[degree($v$, $G_t$)].add($v$)
			\EndFor			
		\EndFor	
		\State $V = V\setminus \{u\}$
		\State{\textbf{return} u}
		\EndProcedure
	\end{algorithmic}
\end{algorithm}

\subsubsection{Solving {\large {\problemmm}}}
For the {\problemmm} problem, we define the score for a node $v$ in $S$, $\textit{score}_m$, as the minimum degree of $v$ in the sequence $\calG\left[S\right]$. That is,
\begin{equation*}\label{eq:mm1}
\textit{score}_m\left(v,\calG\left[S\right]\right) = \min_{G_t\in \calG}\degree\left(v,G_t\left[S\right]\right).
\end{equation*}

At the $i$-th iteration of the $\algobff$ algorithm, we select the node $v_i$ with the minimum $\textit{score}_m$ value.
We call this instantiation of the {\algobff} algorithm  {\algobffmm}.
Below we prove that {\algobffmm} provides the optimal solution
to the {\problemmm} problem.

\begin{proposition}
The {\problemmm} problem can be solved optimally in polynomial time using the
{\algobffmm} algorithm.
\label{prop:mm}
\end{proposition}

\begin{proof}
Let $i$ be the iteration
of the {\algobffmm} algorithm, where for the first time, a node that belongs to an optimal solution
$S^\ast$ is selected to be removed.
Let $v_i$ be this node.
 Then clearly, $S^\ast\subseteq S_{i-1}$ and by the fact that at every iteration we remove edges from the graphs we have that
\[
\textit{score}_m\left(v_i,\calG\left[S_{t-1}\right]\right)\geq \textit{score}_m\left(v_i,\calG\left[S^\ast\right]\right).
\]
Since $v_i$ is the node we pick at iteration $i$, every node $u\in S_{i-1}$
satisfies:
\begin{center}
$\min_{G_t\in\calG}\degree(u,G_t[S_{i-1}]) $ = $ \textit{score}_m\left(u,\calG\left[S_{i-1}\right]\right)$ $\geq$\\
  $\textit{score}_m\left(v_i,\calG\left[S_{i-1}\right]\right) $
$ \geq  \textit{score}_m\left(v_i,\calG\left[S^\ast\right]\right)$.
\end{center}

Since this is true for every node $u$, this means that $S_{i-1}$ is indeed optimal and that our algorithm will find it.
\end{proof}

The running time of {\algobffmm} is $O(n\tau + M)$, where
$n=|V|$, $\tau$ the number of snapshots in the history graph and $M=m_1+m_2+\ldots + m_\tau$ the total number
of edges that appear in all snapshots.
\begin{algorithm}[ht!]
	\caption{\small The score$_a$ algorithm.}
	\label{algo:score_a}
	\begin{algorithmic}[1]
		\small
		\Statex{{\bf Input:} Graph history $\calG = \{G_1, \dots G_\tau\}$}
		\Statex{{\bf Output:} Node with the minimum score$_a$}
		\vspace{0.1cm}
		\hrule
		\vspace{0.1cm}
		\State  $\widehat{H}_\calG \leftarrow$ construct the average graph of $\calG$
        \State ${\cal L}[d]$ $\leftarrow$ list of nodes with degree $d$ in $\widehat{H}_\calG$\\

        \Procedure{ScoreAndUpdate}{$ $}
        	\State $score_a$ $\leftarrow$ smallest $d$ s.t. ${\cal L}[d] \neq \emptyset$
        	\State $u$ = ${\cal L}[score_a]$.get()
			\State ${\cal L}$[$score_a$].remove($u$)
			\For{\textbf{each} $(u, v) \in \widehat{E}$}
				\State ${\cal L}$[degree($v$, $\widehat{H}_\calG)$].remove($v$)
				\State $\widehat{E} = \widehat{E} - (u, v)$ \textit{\,\,\,\,// update $\displaystyle{degree_{v \in V}(v, \widehat{H}_\calG)}$}
				\State ${\cal L}$[degree($v$, $\widehat{H}_\calG$)].add($v$)%
		\EndFor	
		\State  $V = V \setminus \{u\}$
		\State{\textbf{return} u}
		\EndProcedure
	\end{algorithmic}
\end{algorithm}
The node with the minimum  $\textit{score}_m$ value is computed by the procedure {\sc ScoreAndUpdate} shown in Algorithm \ref{algo:score_m}, which also removes the node and its edges from all snapshots.
For each snapshot $G_t$, we keep the list of nodes
with degree $d$ (line 1 in Algorithm \ref{algo:score_m}); these lists can be constructed in time $O(n\tau)$.
Given these lists, the time required to find
the node with the minimum $\textit{score}_m$ is $O(\tau)$ (lines 4--8).
Now in all snapshots, the neighbors of the removed node
need to be moved from their position in the $\tau$ lists (lines 9--14); the degree of every neighbor of the removed node is decreased by one.
Throughout the execution of the algorithm at most $O(M)$ such moves can happen.
Therefore, the total running time of  {\algobffmm} is $O(n\tau + M)$.
Note that an algorithm that iteratively removes from a graph $G$ the node with the minimum degree  was first studied in \cite{DBLP:conf/swat/AsahiroITT96} and shown to compute a 2-approximation of the densest subgraph problem for the $d_a(G)$ density in \cite{DBLP:conf/approx/Charikar00} and  the optimal  for the
$d_m(G)$ density in \cite{DBLP:conf/kdd/SozioG10}.

\subsubsection{Solving {\large {\problemaa}}}
To solve the {\problemaa} problem,
we shall use the average graph $\widehat{H}_{\calG}$  of $\calG$.
Lemma~\ref{lemma:equivalence} shows that $f_\saa(S,\calG) = d_a\left(\widehat{H}_\calG\left[S\right]\right)$.
Thus, based on the results of Charikar~\cite{DBLP:conf/approx/Charikar00} and Goldberg~\cite{goldberg1984finding}, we conclude that:

\begin{proposition}
The {\problemaa} problem can be solved optimally in polynomial time.
\end{proposition}

Although there exists a polynomial-time optimal algorithms for {\problemaa}, the computational complexity of these algorithms (e..g., $O(|V||\widehat{E}|^2)$ for the case of the max-flow algorithm in~\cite{goldberg1984finding}), makes them hard to use for large-scale real graphs.  Therefore, instead of these algorithm we use the {\algobff} algorithm, where  we define the score of a node $v$ in $S$, $\textit{score}_a$, to be equal to its average degree of $v$ in graph history ${\calG}[S]$. That is,
\begin{equation*}\label{eq:aa1}
\textit{score}_a\left(v,\calG\left[S\right]\right) = \frac{1}{\left|\calG\right|}\sum_{G_t\in\calG}\degree\left(v,G_t\left[S\right]\right).
\end{equation*}
At the $i$-th iteration, we select the node $v_i$ with the \emph{minimum average degree in $\calG$}.
We will refer to this instantiation of the {\algobff}, as {\algobffaa}.
Using Lemma~\ref{lemma:equivalence} and the results of Charikar~\cite{DBLP:conf/approx/Charikar00}
we have the following:

\begin{proposition}
{\algobffaa} is a $\frac 12$-approximation algorithm for the {\problemaa} problem.
\end{proposition}

\begin{proof}
It is easy to see that {\algobffaa}  removes
the node with the minimum density in $\widehat{H}_{\calG}$.
Charikar~\cite{DBLP:conf/approx/Charikar00}  has shown that an algorithm that
iteratively removes from a graph the node with minimum density
provides a $\frac 12$-approximation for finding the subset of nodes that maximizes the average density on a single (weighted) graph snapshot.
Given the equivalence we established in Lemma~\ref{lemma:equivalence},
{\algobffaa} is also a $\frac 12$-approximation algorithm for {\problemaa}.
\end{proof}

We show the steps for finding the node with  the minimum  $\textit{score}_a$ value
in Algorithm \ref{algo:score_a} that uses lists of nodes with degree $d$ in the average graph to achieve an $O(n\tau + M)$ total running time for {\algobffaa}.

\subsubsection{Solving {\large{\problemma}} and  {\large {\problemam}}}

We consider the application of {\algobffmm} and {\algobffaa} algorithms for the two problems.
In the following propositions, we prove that the two algorithms give a poor approximation ratio for both problems.
Recall that all our problems are maximization problems, and, therefore, the lower the approximation ratio, the worse the performance of the algorithm.

\begin{proposition}\label{prop:prop1}
	The approximation ratio of algorithm {\algobffmm} for the {\problemam} problem is $O\left(\frac{1}{n}\right)$ where $n$ is the number of nodes.
\end{proposition}

\begin{proof}
In the Appendix.
\end{proof}

\begin{proposition}\label{prop:prop2}
The approximation ratio of algorithm {\algobffaa} for the {\problemam} problem is $O\left(\frac{1}{n}\right)$ where $n$ is the number of nodes.
\end{proposition}

\begin{proof}
In the Appendix.
\end{proof}

\begin{proposition}\label{prop:prop3}
	The approximation ratio of algorithm {\algobffmm} for the {\problemma} problem is $O\left(\frac{1}{\sqrt{n}}\right)$  where $n$ is the number of nodes.
\end{proposition}

\begin{proof}
In the Appendix.
\end{proof}

We also consider applying the {\algobffaa} algorithm that selects to remove the node with the minimum average degree.
We can show that {\algobffaa} has a poor approximation ratio for the {\problemam} problem.

\begin{proposition}\label{prop:prop4}
The approximation ratio of algorithm {\algobffaa} for the {\problemma} problem is $O\left(\frac{1}{\sqrt{n}}\right)$  where $n$ is the number of nodes.
\end{proposition}

\begin{proof}
	In the Appendix.
\end{proof}

The complexity of {\problemma} and {\problemam}  is an open problem.
Jethava and
Beerenwinkel~\cite{jethava15finding} conjecture that the {\problemma} problem is NP-hard, yet they do not provide a proof.

Given that {\algobffaa} and {\algobffmm} have no theoretical guarantees, we also investigate a
\emph{greedy} approach, which selects which node to remove
based on the objective function of the problem at hand.  This greedy approach is again an instance of the iterative algorithm
shown in Algorithm~\ref{algo:bff}. More specifically, for a target function $f$ (either $f_\sam$ or $f_\sma$), given a set $S_{i-1}$, we define the score $\textit{score}_g(v,\calG[S_i])$ of node $v\in S_i$ as follows:
\[
\textit{score}_g(v,\calG[S_{i-1}]) = f\left(S_{i-1},\calG\right) - f\left(S_{i-1}\setminus \{v\},\calG\right).
\]
At iteration $i$, the algorithm selects the node $v_i$ that causes the smallest decrease, or the largest increase in the target function $f$.
We refer to this algorithm as {\algobffgreedy}.
{\algobffgreedy} requires to check all nodes when choosing which node to remove at each step (shown in Algorithm \ref{algo:score_g}), thus leading to complexity $O(n^2\tau + nM)$.

\begin{algorithm}
	\caption{\small The score$_g$ algorithm.}
	\label{algo:score_g}
	\begin{algorithmic}[1]
		\small
		\Statex{{\bf Input:} Graph history $\calG = \{G_1, \dots G_\tau\}$; an aggregate density function $f$}
		\Statex{{\bf Output:} Node with the minimum score$_g$}
		\vspace{0.1cm}
		\hrule
		\vspace{0.1cm}
        \Procedure{ScoreAndUpdate}{$ $}
            \State $score_{g}[u] = \emptyset$ \textbf{for} $u \in V$

        	\For{\textbf{each} $u \in V$}
        		\State $V' = V \setminus \{u\}$
        		\State $score_{g}[u] = f(V',\calG)$
        	\EndFor
        	
        \State $u$ = $arg\displaystyle\min_{v \in V}{score_{g}[v]}$

		\For{\textbf{each} $G_t \in \calG$}
			\For{\textbf{each} $(u, v) \in E_t$}
				\State $E_t = E_t - (u, v)$
			\EndFor	
		\EndFor
		\State $V = V \setminus \{u\}$
		\State{\textbf{return} u}
		\EndProcedure
	\end{algorithmic}
\end{algorithm}

\section{The {\large{\problemtbff}} Problem}
\label{sec:tbff}
In this section, we relax the requirement that the nodes
are connected in all snapshots of a graph history.
Instead, we ask to find the subset of nodes with the maximum aggregate density in at least $k$ of the snapshots.
We call this problem  {\it On-Off} {\problembff} ({\problemtbff}) problem.
We formally define {\problemtbff}, we show that it is NP-hard and develop two general types of algorithms for efficiently solving it in practice.

\subsection{Problem definition}
In the {\problemtbff} problem, we seek to
find a collection $\calC_k$ of $k$ graph snapshots,
and a set of nodes $S\subseteq V$, such that
the subgraphs induced by $S$ in $\calC_k$ have high aggregate density.
Formally, the {\problemtbff} problem is defined as follows:

\begin{problem}[The On-Off {\problembff} ({\problemtbff}) Problem]
\label{pr:tbff}
Given a graph history ${\calG}$ = $\{G_1$, $G_2$, $\dots,$ $G_\tau\}$, an aggregate density function $f$, and an integer $k$,
find a subset of nodes $S\subseteq V$, and a subset $\calC_k$ of $\calG$ of size $k$, such that $f\left(S,\calC_k\right)$
is maximized.
\end{problem}

As for Problem~\ref{pr:bff}, depending on the choice of the aggregate density function
$f$, we have four variants of {\problemtbff}. Thus, $f_\smm$, $f_\sma$, $f_\sam$
and $f_\saa$ give rise to problems {\problemtbffmm}, {\problemtbffma}, {\problemtbffam} and
{\problemtbffaa} respectively.

Note that the subcollection of graphs $\calC_k \subset \calG$ does not need to consist of contiguous graph snapshots.
If this were the case, then the problem could be  solved easily by considering all possible contiguous subsets of
$[1,\tau]$ and
outputting the one with the highest density. However,
all the four variants of the {\problemtbff} become NP-hard if we drop the constraint for consecutive graph
snapshots.

\begin{theorem}
Problem~\ref{pr:tbff} is NP-hard for any definition of the aggregate density function $f$.
\end{theorem}

We will prove that there exists a clique of size at least $k$ in graph $G$ if and only if there exists a set of nodes $S$ and a subset $\calC_k \subseteq \calG$ of $k$ snapshots, with $f(S,\calC_k) \geq 1$. The forward direction is easy; if there exists a subset of nodes $S$ in $G$, with $|S| \geq k$, that form a clique, then selecting this set of nodes $S$, and a subset $\calC_k$ of $k$ snapshots  that correspond to nodes in $S$ will wield $f_\smm(S,\calC_k) = f_\sam(S,\calC_k) = 1$. This follows from the fact that every snapshot is a complete star where $d_m(S,G_i) = 1$ for all $G_i \in \calC_k$. To prove the other direction, we observe that all our snapshots consist of a star graph, and a collection of disconnected nodes. Given a set $S$,
$d_m(S,G_i) = 1$, if $i\in S$ and all nodes in $S$ are connected to the center node $i$, and zero otherwise. Therefore, if $f_\smm(S,\calC_k) = 1$ or $f_\sam(S,\calC_k) = 1$, then this implies that $d_m(S,G_i) = 1$ for all $G_i \in \calC_k$, which means that the $k$ centers of the graph snapshots in $\calC_k$ are connected to all nodes in $S$, and hence to each other. Therefore, they form a clique of size $k$ in the graph $G$.

In the case of $f_\saa$ and $f_\sma$ the construction proceeds as follows: given the graph $G=(V,E)$, with $|E| = m$ edges, we construct a graph history $\calG=\{G_1,\ldots , G_\tau\}$ with $\tau = m$ snapshots. All snapshots are defined over the vertex set $V$. There is a snapshot $G_e$ for each edge $e\in E$, consisting of the single edge $e$.
We can prove that there exists a clique of size at least $k$ in graph $G$ if and only if there exists a set of nodes $S$ and a subset $\calC_K \subseteq \calG$ of $K = k(k-1)/2$ snapshots, with $f(S,\calC_K) \geq 1/k$.

We will prove that there exists a clique of size at least $k$ in graph $G$ if and only if there exists a set of nodes $S$ and a subset $\calC_K \subseteq \calG$ of $K = k(k-1)/2$ snapshots, with $f(S,\calC_K) \geq 1/k$. The forward direction is easy. If there exists a subset of nodes $S$ in $G$, with $|S| = k$, that form a clique, then selecting this set of nodes $S$, and the ${k \choose 2}$ snapshots $\calC_K$ in $\calG$ that correspond to the edges between the nodes in $S$ will yield $f_\saa(S,\calC_K) = f_\sma(S,\calC_K) = 1/k$.

To prove the other direction, assume that there is no clique of size greater or equal to $k$ in $G$. Let $\calC_K$ be any subset of $K = k(k-1)/2$ snapshots, and let $S$ be the union of the endpoints of the edges in $\calC_K$. Since $S$ cannot be a clique, it follows that $|S| = \ell > k$. Therefore, $f_\saa(S,\calC_K) = f_\sma(S,\calC_K)= 1/\ell < 1/k$.

\begin{algorithm}[t]
	\caption{\small The Iterative ({\textsc{ITR}}) {\algotbff}  algorithm.}
	\label{algo:tbff}
		\small
	\begin{algorithmic}[1]
		\Statex{{\bf Input:} Graph history $\calG = \{G_1,\ldots G_\tau\}$; an aggregate-density function $f$; integer $k$}
		\Statex{{\bf Output:} A subset of nodes $S$ and a subset of snapshots $\calC_k\subseteq \calG$.}
		\hrule
		\vspace{0.2cm}
			\State {converged $=$ False}
            \State $(\calC_k^0, S^0) = \textsc{Initialize}\left(\calG,f\right)$
	        \State $ds^0 = 0$
			\While {not converged}
                 \State $\calC_k = \textsc{BestSnapshots}(S^0,f)$
                 \State $S = \algobff(\calC_k, f)$
                 \State $ds = f(S,\calC_k)$
                 \If {$ds<ds^0$}
                 \State Converged $=$ True
                 \Else { $ds^0=ds$, $S^0=S$ }
                 \EndIf
			\EndWhile
			\State \textbf{return} $S,\calC_k$
	\end{algorithmic}
\end{algorithm}

\subsection{{\large {\problemtbff}} algorithms}\label{sec:bfftalgos}
We consider two general types of algorithms:  iterative and incremental algorithms.
The \textit{iterative algorithm} starts with an initial size $k$ collection  $\calC_k$ of graph snapshots and improves it, whereas
the \textit{incremental algorithm} builds the collection incrementally, adding one snapshot at a time.
Next, we describe these two types of algorithms in detail.

Note that depending on whether we are solving the {\problemtbffmm},
{\problemtbffma}, {\problemtbffam} or {\problemtbffaa} problem, we use the appropriate version of the {\algobff} algorithm in each of these algorithms.

\subsubsection{Iterative Algorithm}
\label{sec:bff-k}
The iterative ({\itg}) algorithm (shown in Algorithm \ref{algo:tbff}) starts with an initial collection of snapshots $\calC^0_k$ and set of nodes $S^0$ (routine {\sc Initialize}).
%
At each iteration, given a set $S$,
it finds the $k$ graph snapshots with the highest $d(S, G_i)$ score; this is done by  {\sc BestSnapshots}.
{\sc BestSnapshots} computes the density $d(S, G_i)$ of $S$ in each snapshot $G_i$ $\in$ $\calG$ and outputs the $k$ snapshots $\calC_k$ with the highest density.
Given $\calC_k$, the algorithm then finds the set $S\subseteq V$ such that $f\left(S,\calC_k\right)$ is maximized.
This step essentially solves Problem~\ref{pr:bff}
on input $\calC_k$ for aggregate density function $f$ using the {\algobff} algorithm.
The {\itg} algorithm keeps iterating between collections $\calC_k$ and dense sets of nodes $S$ until
no further iterations
can improve the score $f\left(S,\calC_k\right)$.


An important step of the Iterative {\algotbff} is the initialization of $\calC^0_k$  and $S^0$.
We consider three different alternatives
for this initialization: \emph{random}, \emph{contiguous}, and \emph{at least-k}.

\noindent {\textit {Random initialization}}  ({\itr}): In this initialization, we randomly pick $k$ snapshots
$\calC^0_k$ from $\calG$. These snapshots are then used for solving the corresponding {\problembff} problem
on input $\calC^0_k$ and produce $S^0 = \algobff (\calC^0_k, f)$.

\noindent {\textit {Contiguous initialization}} ({\itc}): In this initialization,
we first find an $S^0$ that consists of  the best $k$  contiguous graph  snapshots.
Given $\calG=\{G_1,\dots ,G_\tau\}$, we go over all the $O(\tau)$
contiguous sets of $k$ snapshots from $\calG$,  and find the set of $k$ snapshots
$\calC^0_k$
and corresponding set of nodes $S^0$ that maximize
$f(S^0,\calC^0_k)$.
The intuition behind this initialization technique is that it assumes that the best $k$ snapshots
of $\calG$ are going to be contiguous. Our experiments demonstrate that in practice this is true
in many datasets -- e.g., in collaboration networks that evolve over time and we expect to see
some temporal locality.

\noindent {\textit {At least-$k$ initialization}} ({\itk}): In this initialization, we solve the {\problembff} problem independently in each snapshot
$G_i\in\calG$. This results in $\tau$ different sets $S_i\subseteq V$, one for each solution of {\problembff} on $G_i$.
$S^0$ includes the nodes that appear in at least $k$ of the $\tau$ sets $S_i$.
The intuition behind this initialization is to include in the initial solution those nodes that appear to be densely connected in many snapshots.
We also experimented with other natural alternatives, such as the union: $S^0 = \cup_{i=1\ldots \tau}S_i$ and
the intersection: $S^0 = \cap_{i=1\ldots \tau}S_i$;  the at least-$k$ approach seems to strike a balance between the two.


The running time of the iterative {\algotbff} algorithm is $O\left(I\left(n\tau+M\right)\right)$, where $I$
is the number of iterations required until convergence, and the $O(n\tau+M)$ comes from the running time of {\algobff}. In practice, we observed that
the algorithm converges in at most $6$ iterations.

\begin{algorithm}[t]
	\caption{\footnotesize The Incremental Density ({\ind}) {\algotbff} algorithm.}
	\label{algo:tbff_incrDens}
	\begin{algorithmic}[1]
	\small
		\Statex{{\bf Input:} Graph history $\calG = \{G_1,\ldots G_\tau\}$; aggregate-density function $f$; integer $k$}
		\Statex{{\bf Output:} A subset of nodes $S$ and a subset of snapshots $\calC_k\subseteq \calG$.}
		\hrule
		\vspace{0.1cm}		
		\State $S_{ij} = \algobff(\{G_i,G_j\}, f)$, $\forall G_i,G_j \in {\calG}$
		\State $\calC_2 = \arg\displaystyle\max_{G_i,G_j \in {\calG}}{f(S_{ij},\{G_i,G_j\})}$
		\For{$i=3$ ; $i \leq k$}
			\For{\textbf{each} $G_t \in {\calG} \setminus {\calC}_{i-1}$}			
				\State $S_t = \algobff(\calC_{i-1} \cup \{G_t\},f)$
			\EndFor
        \State $G_m = \arg\displaystyle\max_{G_t}{f(S_t,\calC_{i-1} \cup \{G_t\})}$		
		\State ${\calC}_i = \calC_{i-1} \cup \{G_m\}$		
		\EndFor		
		\State $S = \algobff(\calC_k,f)$			
		\State \textbf{return} $S,\calC_k$
	\end{algorithmic}
\end{algorithm}

\begin{algorithm}[t]
	\caption{\footnotesize The Incremental Overlap ({\ino}) {\algotbff} algorithm.}
	\label{algo:tbff_incrOver}
	\begin{algorithmic}[1]
	\small
		\Statex{{\bf Input:} Graph history $\calG = \{G_1,\ldots G_\tau\}$; aggregate-density function $f$; integer $k$}
		\Statex{{\bf Output:} A subset of nodes $S$ and a subset of snapshots $\calC_k\subseteq \calG$.}
		\hrule
		\vspace{0.1cm}	
		\State $S_i = \algobff(G_i, f)$, $\forall G_i \in {\calG}$
		\State $\calC_2 = \arg\displaystyle\max_{\substack{G_i,G_j \in {\cal G}}}{\frac{|S_i \cap S_j|}{|S_i \cup S_j|}}$
		\For{$i=3$ ; $i \leq k$}
			\State $S_C = \algobff(\calC_{i-1}, f)$
		\State $G_m = \arg\displaystyle\max_{G_t}{\frac{|S_t \cap S_C|}{|S_t \cup S_C|}}$
		\State ${\calC}_i = \calC_{i-1} \cup \{G_m\}$			
		\EndFor
		\State $S = \algobff(\calC_k,f)$			
		\State \textbf{return} $S,\calC_k$	
	\end{algorithmic}
\end{algorithm}

\subsubsection{Incremental Algorithm}
The incremental algorithm starts with a collection $\calC_2$ with two snapshots and  incrementally adds snapshots
to it until a collection $\calC_k$ with $k$ snapshots is formed.
Then, the  appropriate {\algobff} algorithm is used to compute the most dense subset of nodes $S$ in  $\calC_k$.

We use two different policies for selecting snapshots. The first one, termed \textit{incremental density} ({\ind}) algorithm (shown in Algorithm \ref{algo:tbff_incrDens}), selects graph snapshots so as to maximize density, whereas the second one, termed \textit{incremental overlap} ({\ino}) algorithm (shown in Algorithm \ref{algo:tbff_incrOver}), selects graph snapshots so as to maximize the overlap among nodes in the dense subsets.

\vspace*{0.05in}
\noindent {\textit {Incremental density}} ({\ind}):
To select the pair of snapshots to form the initial collection $\calC_2$,  we solve the {\problembff} problem independently for each pair of
snapshots
$G_i,G_j\in\calG$. This gives us  ${\tau}\choose{2}$ dense sets $S_{ij}$ as solutions. We select the pair of snapshots whose dense subgraph $S_{ij}$ has the largest density (lines 1--2).
The algorithm
then builds the solution incrementally in iterations. In iteration $i$, we construct the solution $\calC_i$ by adding to solution $\calC_{i-1}$ the graph snapshot $G_m$ that maximizes the density function $f$. That is if $S_t$ is the densest subset in the sequence $\calC_{i-1}\cup\{G_t\}$, $G_m = \displaystyle\arg\max_{G_t} f(S_t,\calC_{i-1}\cup\{G_t\})$ (lines 3--6).
The running time of the {\ind} algorithm is $O\left(\tau^2(n+M) + k\tau\left(kn\tau+M\right)\right)$.
The first term is due to the initialization step in line 1, where we look for the best pair of snapshots.
If efficiency is important we can initialize the algorithm with a random pair to save time.

\noindent {\textit {Incremental overlap}} ({\ino}):
To form the initial collection $\calC_2$,  we first solve the {\problembff} problem independently in each snapshot
$G_i\in\calG$. This gives us  $\tau$ different sets $S_i\subseteq V$, where  $S_i$ is the most dense subgraph in  $G_i$.
The algorithm selects from these $\tau$ sets  the two most similar ones,  $S_i$ and $S_j$, and initializes $C_2$ with the corresponding snapshots $G_i$ and $G_j$ (lines 1--2). For defining similarity
between sets of nodes, we use the Jaccard similarity.
To form $C_i$ from $C_{i-1}$, the algorithm first solves the {\problembff} problem in $C_{i-1}$. Let $S_C$ be the solution.
Then, it selects from the remaining snapshots and adds to $C_{i-1}$ the snapshot $G_m$ whose dense set $S_t$ is the most similar  with  $S_C$ (lines 3--6).
%
The running time of the {\ino} algorithm is $O(\tau^2n + k\left(n\tau+M\right))$,
where the first term is the time for the initialization and the second term for the for-loop.


\section{ {\Large {\problembff}} Problem Extensions}\label{sec:generalized}

The definitions of the {\problembff} and {\problemtbff} problem focus on the identification of a set of nodes $S$ such that their
aggregate density is maximized.
We now consider natural extensions of the {\problembff} problem by placing additional constraints on the dense subgraphs.

\spara{Query-node constraint:} An interesting extension is introducing a set $Q$ of \emph{seed query nodes}
and requiring that the output set of nodes
$S$ has high density and also contains the input seed nodes.
A similar extension was introduced for static (e.g., single snapshots) graphs in ~\cite{DBLP:conf/kdd/SozioG10}.
In practice, this variant of  {\problembff}  identifies the lasting ``best friends"
of the query nodes. We call this the {\problemqbff} problem.

We can modify the {\algobff} algorithms appropriately so that they take into consideration this additional constraint.
In particular, {\algobffmm}  stops when a query node in $Q$ is selected to be removed.
Let us call this modified algorithm, {\algoQbffmm}.  We can prove the following proposition. (We omit the proof due to space constraints.)

\begin{proposition}
{\algoQbffmm} solves the {\problemQmm} problem  optimally in polynomial time.
\end{proposition}

We also modify {\algobffaa} so that it does not remove seed nodes as follows: 
If at any step, the node with the minimum average degree happens to be a seed node, the algorithm selects to remove the node with the next smallest degree that is not a seed node. The algorithm stops when the only remaining nodes are seed nodes.
Let us call this modified algorithm, {\algoQbffaa}.

\begin{proposition}
Let $S^*$ be an optimal solution  for the {\problemQaa} problem and $S_A$ be the solution of the
{\algoQbffaa} algorithm.
It holds:
$f_\saa(S_A)$ $\geq$ $\frac{s \, f_\saa(S^*) + 2 \, \omega} {2(s+q)}$, where $q$ = $|Q|$, $s$ = $|S^* \setminus Q|$ and $\omega$ $=$ $\sum_{u \in Q}degree(u, S^*)$.
\end{proposition}

\begin{proof}
By Lemma \ref{lemma:equivalence}, it suffices to show that the {\algoQbffaa} algorithm provides an approximation of the 
average density of a single graph $G$.
Let $S^*$ be the optimal solution for $G$. 
Let  $G'$ be the graph that results from $G$ when we delete all edges between two query nodes in $G$.
Clearly,  $S^*$ is also an optimal solution for $G'$.
Assume that we assign each edge $(u, v)$ to either $u$ or $v$. For each node $u$, let $a(u)$ be	the number of
edges assigned to it and let $a_{max}$ = $max_u\{a(u)\}$.
It is easy to see that $f_\saa(S^*)$ $\leq$ $\frac{1}{2}$ $a_{max}$, since each edge in the optimal solution must be assigned to a node in it.
Now assume that the assignment of edges to nodes is performed as the {\algoQbffaa} algorithm proceeds.
Initially, all edges are unassigned. When at step $i$, a node $u$ is deleted, we assign to  $u$ all the edges that go from 
$S_{i-1}$ to $u$. Note that this assignment maintains the invariant that at each step, all edges between two nodes in the current set $S$
are unassigned, while all other edges are assigned. When the algorithm stops, all edges have been assigned.
Consider a single iteration of the algorithm when a node $u_{min}$ is selected to be removed and let $S$ be the current set.
Let $s$ be the number of non-query nodes in $S$, and $q$ = $|Q|$ be the number of query nodes. 
It holds: $f_\saa(S)$ $=$ $\frac{1}{s+q}$$\sum_{u \in S}$$degree(u)$ = $\frac{1}{s+q}$$\sum_{v \in S \setminus Q }$$degree(u)$+$\frac{1}{s+q}$$\sum_{u \in Q}$$degree(u)$.
Let $\Omega$ $=$ $\frac{1}{s+q}$$\sum_{u \in Q}$$degree(u)$.
Since  $u_{min}$ has the smallest degree among all nodes in $S$ but the seed nodes, 
we have $f_\saa(S)$  $\geq$ $\frac{1}{s+q}$ $s$\,$a(u_{min})$ + $\Omega$. Since all edges are assigned and edges are assigned to a node only when this node is removed, at some step of the execution of the algorithm $a(u_{min})$ = $a_{max}$. Thus, for some $S$,
$f_\saa(S)$ $\geq$  $\frac{s}{s+q}$$a_{max}$$+$$\Omega$ $\geq$  $\frac{s}{s+q}$ $\frac{1}{2}$$f_\saa(S^*)$$+$$\Omega$.
\end{proof}

\spara{Connectivity constraint:} Another meaningful extension is to impose restrictions on the connectivity of $S$.  The connectivity
of $S$ in a graph history $\calG=\{G_1,\ldots , G_\tau\}$ may have many different interpretations.
One may consider a version where all the induced subgraphs $G_t[S]$ for $t\in\{1,\ldots ,\tau\}$
are connected. Another alternative is that at least $m$ $>$ 0 of the $\tau$ $G_t[S]$'s are
connected.  Here, we assume that a definition of connectivity for $S$ is given in the form of a predicate $connected(S, \calG)$ which is true if $S$ is connected and false otherwise.
Our problem now becomes: given a  graph history $\calG$,
	a set of query nodes $Q$ $\subset$ $V$ 
	and an aggregate density function $f$, find a subset of nodes $S\subseteq V$, 
	such that  (1) $f(S,\calG)$ is maximized, (2) $Q$ $\subseteq$ $S$, and (3) $connected(S, \calG)$ is true.

To solve this problem, 
we can modify {\algobff} so that it tests for the connectivity predicate
$connected(S, \calG)$ and stops when 
the connectivity constraint no longer holds.
In our experiments, we apply a simplest test, just running the algorithms on the connected components of
the query nodes.

\spara{Size constraint:}
Finally, note that the definition of  {\problembff}  does not impose a constraint on the size
of the output set of nodes $S$. In that respect the problem is parameter-free. 
If necessary, one can add
an additional constraint to the problem definition
by imposing a cardinality constraint on the output $S$. 
However the cardinality constraint makes the subgraph-discovery problem computationally hard~\cite{DBLP:conf/kdd/SozioG10}. This also holds for the {\problembff} problem; simply
consider a graph history with replicas of the same single-snapshot graph.
\section{Experimental Evaluation}
\label{sec:experiments}
The goal of our experimental evaluation is threefold.
First, we want to evaluate the performance of our algorithms for the {\problembff} and the {\problemtbff} problems in terms of the quality of the solutions and running time. Second, we want to compare the different variants of the aggregate density functions.
Third, we want to show the usefulness of the problem, by
presenting results of {\problembff}'s and {\problemtbff}'s in two real datasets, namely research collaborators in {\db} and hashtags in Twitter.

\begin{table}[ht!]
\centering
\caption{Real dataset characteristics}
\begin{tabular}{c c c c c }
\hline
{\textbf{Dataset}} & \# Nodes & \# Edges (aver. per snapshot) & \# Snapshots\\
\hline
{\dbten} &2,625 & 1,143 &  10 \\
{\orf} & 11,492 & 22,569  & 9 \\
{\ort} & 11,806 & 31,559  & 9 \\
{\cai} & 31,379 & 45,833  & 122 \\
{\twitter} & 849 & 100  & 15 \\
{\as} & 7,716 & 7,783  & 733 \\
\hline
\end{tabular}
\label{table:real}
\end{table}

\noindent \textit{\textbf{Datasets and setting.}}
To evaluate our algorithms, we use a number of real graph histories, where the snapshots correspond to
collaboration, computer, and concept networks.
\squishlist
\item The {\dbten}\footnote{http://dblp.uni-trier.de/} dataset
contains yearly snapshots of the co-authorship graph in the 2006-2015 interval, for 11 top database and data mining
conferences.
There is an edge between two authors in a
graph snapshot, if they co-authored a paper in the corresponding year and more than two papers in the corresponding interval.

\item The {\orf}\footnote{https://snap.stanford.edu/data/oregon1.html} dataset consists of nine graph snapshots of AS peering information inferred from Oregon route-views between March 31 2001 and May 26 2001 (one snapshot per week).

\item The {\ort}\footnote{https://snap.stanford.edu/data/oregon2.html} dataset consists of nine weekly snapshots of AS graphs, between March 31, 2001 and May 26, 2001.

\item The {\cai}\footnote{http://www.caida.org/data/as-relationships/} dataset, contains
122 CAIDA autonomous systems (AS) graphs, derived from a set of route views BGP-table instances.

\item In the {\twitter} dataset \cite{DBLP:conf/edbt/Tsantarliotis15},  nodes are hashtags of tweets and edges represent the co-appearance of hashtags in a tweet.
The dataset contains 15 daily snapshots from October 27, 2013 to November 10, 2013.

\item The {\as}\footnote{\small https://snap.stanford.edu/data/as.html} dataset represents a communication network of who-talks-to-whom from the BGP (Border Gateway Protocol) logs.  The dataset contains 733 daily snapshots which span an interval of 785 days from November 8, 1997 to January 2, 2000.
\squishend
The  dataset characteristics are summarized in Table \ref{table:real}.

Since we do not have any ground truth information for the real
datasets, we  also use synthetic datasets.
In particular, we create graph snapshots using the forest fire model
\cite{forestfire}, a well-known model for creating evolving networks, using the default forward and backward burning probabilities of 0.35.
Then, we plant dense subgraphs in these snapshots, by randomly selecting a set $X$ $\subset$ $V$ of the nodes and creating additional edges between them, different at each snapshot.

We ran our experiments on a system with a quad-core Intel Core i7-3820 3.6 GHz processor,
with 64 GB memory. We only used one core in all experiments.

\subsection{{\large {\problembff}} evaluation}

Since, as shown in Section~\ref{sec:bffalgos},  {\algobffmm} and  {\algobffaa}
are provably good for the {\problemmm} and {\problemaa} problems
respectively, we only consider these  algorithms for these problems.
For the {\problemma} and {\problemam} problems, we use all three algorithms, i.e.,
{\algobffmm}, {\algobffaa},  and {\algobffgreedy}.
For the {\problemma} problem, we also use the \textit{DCS} algorithm proposed in ~\cite{jethava15finding} for a  problem similar to {\problemma}.
The \textit{DCS} algorithm is also an iterative algorithm that removes nodes, one at a time.
At each step, \textit{DCS} finds the subgraphs with the largest average density for each of the snapshots.
Then, it identifies the subgraph with the smallest average density among them and removes the node that
has the smallest degree in this subgraph.

\spara{Accuracy of {\algobff} and comparison of the density definitions:}
We start by an evaluation of the accuracy of our algorithms along with a comparison of the different aggregate densities.
Since we do not have any ground truth information for the real
data, we use first the synthetic datasets.

First, we create 10 graph snapshots with $4,000$ nodes each using the forest fire model \cite{forestfire}.
Then, in each one of the 10 snapshots we plant
a dense random subgraph $A$ with $100$ nodes by inserting extra edges with  probability $p_A$.
We vary the edge probabilities from  $p_A = 0.1$ to $p_A = 0.9$, and
in Fig.~\ref{fig:f1-synth}(a), we report the $F$ measure achieved for the four density definitions, when trying to recover subgraph $A$.
Recall that the $F$ takes values in $[0,1]$ and the larger the value the better the
recall and precision of the solution with respect to the ground truth (in this case $A$).
{\problemmm} is the most sensitive measure, since it reports $A$ as a dense subgraph even for the smallest edge probability.
{\problemma} and {\problemam} achieve a perfect $F$ value,
for an edge probability larger than $p_A = 0.1$ and {\problemaa} for an edge probability at least $p_A = 0.3$.
For smaller values, these three density definitions locate supersets of $A$, due to averaging.
All variations of the {\algobff} algorithms produce the same results.
We now study how the various density definitions behave when there is a second dense subgraph.
In this case, we plant a subgraph $A$ with edge probability  $p_A = 0.5$ in all snapshots and
a second dense subgraph $B$ with the same number of nodes as $A$
and edge probability $p_B = 0.9$ in a percentage
$\ell$ of the snapshots, for different values of $\ell$.
Fig.~\ref{fig:f1-synth}(b) depicts which of two graphs, graph $A$ (shown in blue), or graph $B$ (shown in red), is output by the {\algobff} algorithms for the different density definitions.
{\problemmm} and {\problemma} report $A$ as the densest subgraph, since these measures
ask for high density at each and every snapshot.
However,  {\problemam} and {\problemaa} report $B$, when the very dense subgraph $B$ appears in a sufficient number (more than half) of the snapshots.
All density definitions and algorithms, recover the exact set $A$, or $B$, at each case.

\begin{figure}[ht!]
\centering
\resizebox{1.\columnwidth}{!}{
\begin{tabular}{cc}
\subfloat[$F$-measure for the $BFF$ problems]{\includegraphics[width=0.5\columnwidth]{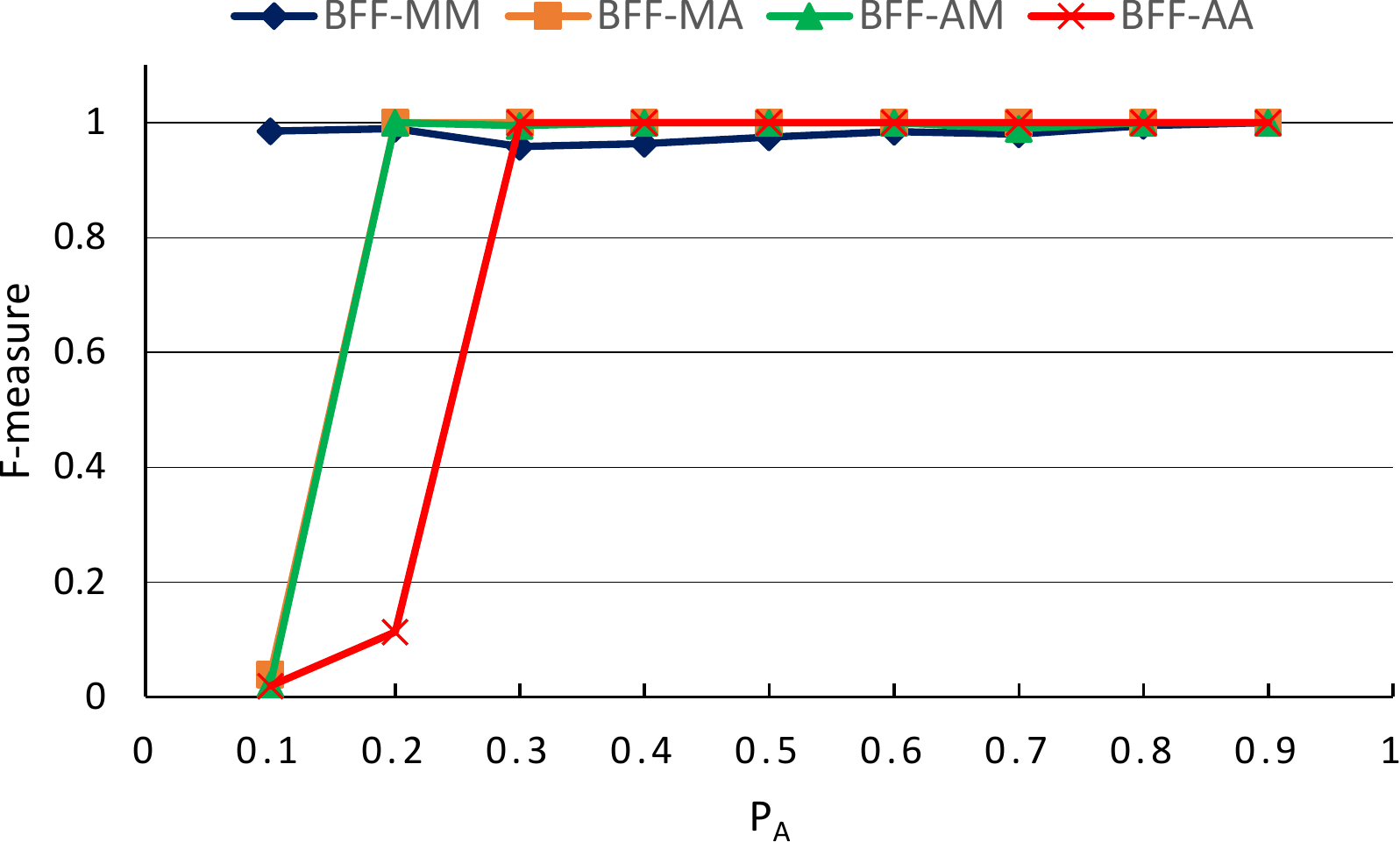}}
   & \subfloat[Dense graph reported by each of the $BFF$ problems]{\includegraphics[width=0.52\columnwidth]{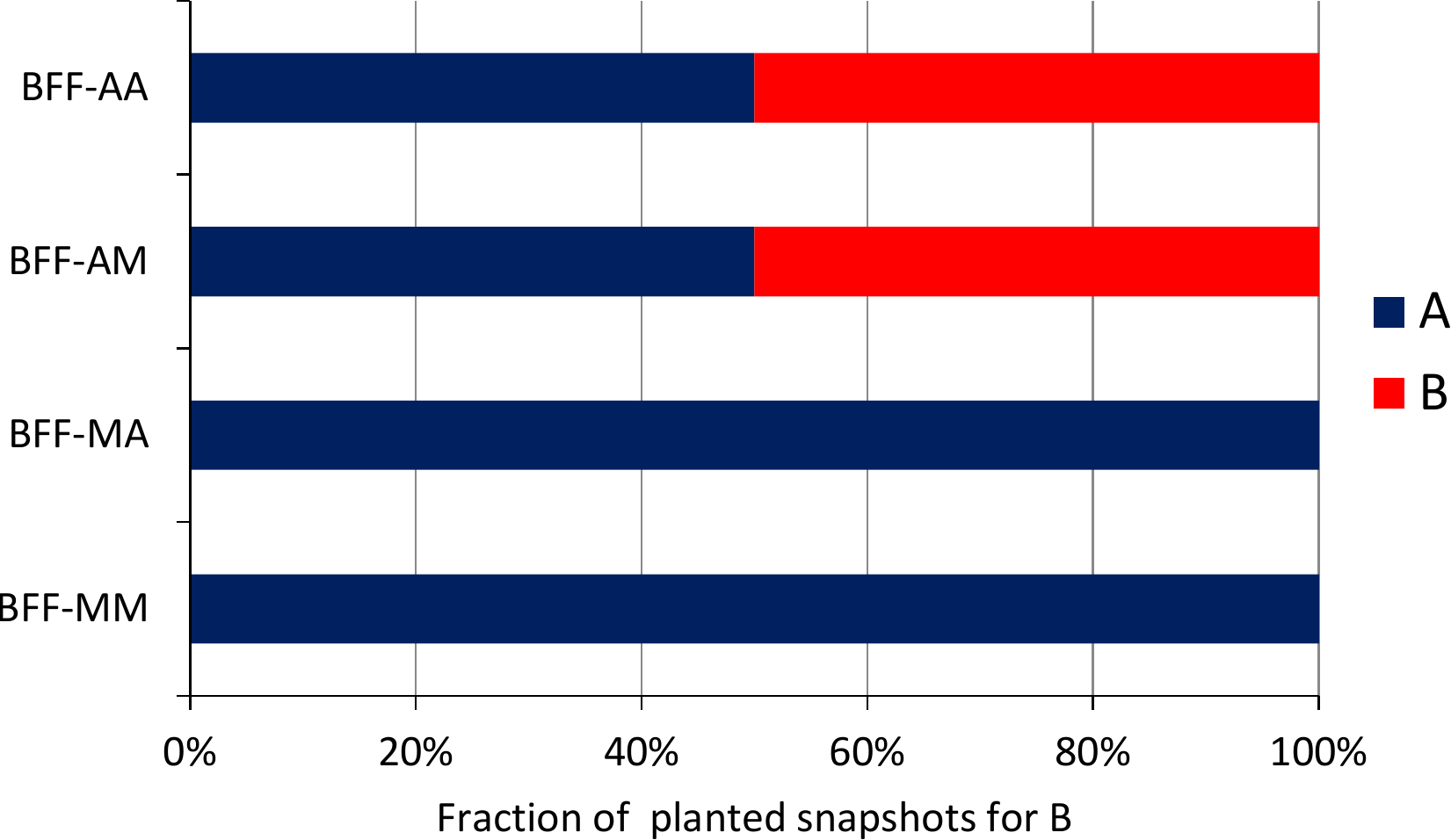}}
\end{tabular}
}
\caption{Accuracy and density definition comparison for {\problembff}}
\vspace{-0.3cm}
\label{fig:f1-synth}
\end{figure}

\begin{table*}[ht!]
\centering
\caption{Results of the different algorithms for the {\problembff} problem on the real datasets.}
\resizebox{1.0\linewidth}{!}{
\begin{tabular}{|c| c| c|| c | c|| c | c || c| c|| c| c|| c| c|| c| c|| c|c||c |c|c|c|}
\hline
\multirow{3}{*}{Datasets} & \multicolumn{2}{|c||}{\small {\problemmm}} & \multicolumn{8}{|c||}{\small {\problemma}} & \multicolumn{6}{|c||}{\small {\problemam}} & \multicolumn{2}{|c|}{\small {\problemaa}} \\ \cline{2-19}
& \multicolumn{2}{|c||}{\small \algobffmm} & \multicolumn{2}{|c||}{\small \algobffmm} & \multicolumn{2}{|c||}{\small \algobffaa} &
\multicolumn{2}{|c||}{\small \algobffgreedy} &  \multicolumn{2}{|c||}{\small {DCS}} &  \multicolumn{2}{|c||}{\small \algobffmm} & \multicolumn{2}{|c||}{\small \algobffaa} & \multicolumn{2}{|c||}{\small \algobffgreedy} & \multicolumn{2}{|c|}{\small \algobffaa}\\ \cline{2-19}
& {\small Size} & {\small $f_{mm}$} & {\small Size} & {\small $f_{ma}$} & {\small Size} & {\small $f_{ma}$} & {\small Size} & {\small $f_{ma}$} & {\small Size} & {\small $f_{ma}$} & {\small Size} & {\small $f_{am}$} & {\small Size} & {\small $f_{am}$} & {\small Size} & {\small $f_{am}$} & {\small Size} & {\small $f_{aa}$}\\
 \hline
{\dbten} & 11 & 1.0 & 3 & 1.33 & 8 & 1.75 & 61 & 1.7 & 14 & 1.29 & 11 & 1.0 & 4 & 1.7 & 4 & 1.0 & 8 & 2.75\\
{\orf} & 33 & 14.0 & 80 & 23.7 & 73 & 23.86 & 80 & 24.05 & 77 & 24.05 & 33 & 14.22 & 35 & 15.0 & 20 & 2.0 & 59 & 25.73\\
{\ort} & 75 & 23.0 & 140 & 44.33 & 131 & 45.24 & 132 & 45.95 & 116 & 44.91 & 63 & 24.44 & 44 & 23.22 & 461 & 3.22 & 147 & 47.89\\
{\cai} & 17 & 8.0 & 33 & 13.76 & 29 & 12.76 & 60 & 15.43 & 57 & 15.05 & 20 & 12.72 & 36 & 18.11 & 311 & 3.43 & 96 & 33.21\\
{\twitter} & - & 0.0 & 836 & 0.04 & 7 & 0.29 & 13 & 0.62 & 720 & 0.05 & - & 0.0 & 3 & 1.0 & 3 & 1.0 & 5 & 1.38\\
{\as} & 15 & 4.0 & 19 & 8.53 & 18 & 6.67 & 20 & 9.0 & 16 & 8.75 & 12 & 7.44 & 14 & 9.05 & 14 & 3.14 & 38 & 16.38\\
\hline
\end{tabular}}
\label{table:size-dens}
\end{table*}

\begin{table*}[ht!]
\centering
\caption{Execution time (sec) of the different algorithms for the {\problembff} problem on the real datasets.}
\resizebox{0.9\linewidth}{!}{
\begin{tabular}{|c|c||c|c|c|c||c|c|c||c|}
\hline
\multirow{2}{*}{Datasets} & \small {\problemmm}  & \multicolumn{4}{c||}{\small {\problemma}}    & \multicolumn{3}{c||}{\small {\problemam}} & \small {\problemaa}  \\ \cline{2-10}
                          & {\small \algobffmm}& {\small \algobffmm} & {\small \algobffaa} & {\small \algobffgreedy} & {\small {DCS}} & {\small \algobffmm} &  {\small \algobffaa} & {\small \algobffgreedy} & {\small \algobffaa} \\ \hline
{\dbten} & 0.08 &  0.05 &   0.03   &  2.04  & 0.34   & 0.05  & 0.08 & 1.58  & 0.04\\ \hline
{\orf} & 0,27 & 0.24 & 0.21 & 48 & 0.83 & 0.48 & 0.57 & 131 & 0.28 \\ \hline
{\ort} & 0.36 & 0.29 & 0.47 & 52 & 1.03 & 0.58 & 0.65 & 57 & 0.48 \\ \hline
{\cai} & 2.24 & 2.51 & 2.30 & 2,519 & 11.22 & 6.31 & 5.97 & 1,652 & 2.14  \\ \hline
{\twitter} & 0.37 & 0.57 & 0.24 & 2.81 & 0.47 & 0.85 & 0.28 & 2.65 & 0.52  \\ \hline
{\as} & 3.49 & 2.82 & 2.16 & 738 & 17.37 & 9.29 & 10.43 & 470 & 2.64  \\ \hline
\end{tabular}}
\label{table:time-dens}
\end{table*}

We also run all algorithms using the real datasets and present the results in
Table \ref{table:size-dens}, where
we report the value of the objective function
and the size of the solution.
A first observation is that as expected, the value of the aggregate density of the reported solution (independently of the problem variant) increases with the density of the graphs.
For {\problemmm} problem we observe that
the solutions usually have small cardinality compared to the solutions for other problems,
since the $f_{\smm}$ objective is rather strict (the solution for {\twitter} was empty).
The solutions for  {\problemmm}  problem in the
autonomous-system datasets
appear to have higher $f_\smm$ scores.
This may be due to the fact that there are larger groups of nodes with lasting connections
in these datasets, e.g., nodes that communicate intensely between each other during the
observation period.

\spara{Comparison of  {\algobff} alternatives for  {\problemma} and   {\problemam}:}
As shown in Table~\ref{table:size-dens}, for the {\problemma} problem,  {\algobffgreedy} and {\algobffaa} perform overall the best in all datasets producing subgraphs with large $f_{\sma}$ values. {\algobffaa}  performs slightly worse than {\algobffgreedy} only in the {\cai} dataset. In the {\cai} dataset, due probably to the large number of snapshots,
{\algobffaa} -- which is based on the average degree -- returns a set with the smallest density.
{\algobffmm} and \textit{DCS} have comparable performance, since they both remove nodes with small degrees in individual snapshots. They are both outperformed by {\algobffaa} and {\algobffgreedy}.

For the {\problemam} problem,
 {\algobffaa} outperforms both {\algobffmm} and {\algobffgreedy}.   Our deeper analysis of the
inferior performance of {\algobffgreedy} for this problem revealed that
{\algobffgreedy} often gets trapped in local maxima after removing just a few nodes of the graph
and it cannot find good solutions.


\spara{Running time:}
In Table \ref{table:time-dens}, we report  execution times.
As expected, the response time of {\algobffgreedy} algorithm is the slowest in all datasets, due to its quadratic complexity.
For the {\problemma} problem, $\algobffaa$ is in general faster than {\sc DCS}.
The difference in times in {\algobffmm} algorithms are due to differences in the computation of the density functions.
Additional experiments including ones with  synthetic datasets with larger graphs and more intervals that show similar behavior are depicted in Figs. \ref{fig:synth_times}(a)(b). 
In particular the Fig. \ref{fig:synth_times}(a) show the execution time of the different algorithms for the {\problemma} problem for varying nodes with $\tau = 10$, whereas Fig. \ref{fig:synth_times}(b) shows the execution time for varying snapshots.

\spara{Summary:}
In conclusion, our algorithms successfully discovered the planted dense subgraphs even when their density is small, with {\problemmm} being the most sensitive measure.
Minimum aggregation over densities (i.e., {\problemmm}, {\problemma}) requires a dense subgraph to be present at all snapshots,
whereas average aggregation over densities  (i.e., {\problemam}, {\problemaa}) asks that the nodes are sufficiently connected with each other on average.
For the {\problemma} and {\problemam}  problems,  $\algobffaa$ returns in general more dense subgraphs than the alternatives (including {\sc DCS}).
Both $\algobffaa$ and $\algobffmm$ scale well. They perform similarly  for the different density functions with the
differences in running time attributed to the complexity of calculating the respective functions.

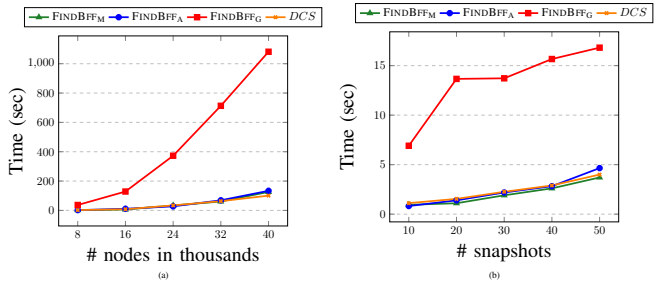
\begin{figure}
\centering
\resizebox{1.0\columnwidth}{!}{
\begin{tabular}{cc}
\subfloat[]
{\begin{tikzpicture}
\begin{axis}[
	xtick = {8, 16, 24, 32, 40},
	xlabel={\# nodes in thousands},
	ylabel={Time (sec)},
	y label style={font=\LARGE},
    x label style={font=\LARGE},
	legend style={legend columns=4, anchor=north,at={(0.50,1.15)}},
	ymajorgrids=true,
	grid style=dashed,
]
	\addplot[color=darkgreen, mark=triangle*,line width=1.5pt,mark size=2pt]coordinates{(8,1.822)(16,6.068)(24,33.545)(32,60.515)(40,126.063)};
	\addplot[color=blue, mark=*,line width=1.5pt,mark size=2pt]coordinates{(8,1.898)(16,11.152)(24,27.398)(32,68.918)(40,133.601)};
	\addplot[color=red, mark=square*,line width=1.5pt,mark size=2pt]coordinates{(8,36.863)(16,128.57)(24,372.646)(32,713)(40,1081.589)};
\addplot[color=orange, mark=x,line width=1.5pt,mark size=2pt]coordinates{(8,2.225)(16,8.047)(24,34.160)(32,63)(40,100.569)};
\legend{{\algobffmm},{\algobffaa},{\algobffgreedy},$DCS$}
\end{axis}
\end{tikzpicture}}
&
\subfloat[]
{\begin{tikzpicture}
\begin{axis}[
	xtick = {10, 20, 30, 40, 50},
	xlabel={\# snapshots},
	ylabel={Time (sec)},
	y label style={font=\LARGE},
    x label style={font=\LARGE},
    scaled y ticks = false,
	legend style={legend columns=4, anchor=north,at={(0.5,1.15)}},
	ymajorgrids=true,
	grid style=dashed,
]
	\addplot[color=darkgreen, mark=triangle*,line width=1.5pt,mark size=2pt]coordinates{(10,0.942)(20,1.101)(30,1.901)(40,2.603)(50,3.704)};
	\addplot[color=blue, mark=*,line width=1.5pt,mark size=2pt]coordinates{(10,0.818)(20,1.394)(30,2.186)(40,2.822)(50,4.650)};
	\addplot[color=red, mark=square*,line width=1.5pt,mark size=2pt]coordinates{(10,6.908)(20,13.661)(30,13.720)(40,15.667)(50,16.816)};
\addplot[color=orange, mark=x,line width=1.5pt,mark size=2pt]coordinates{(10,1.112)(20,1.543)(30,2.257)(40,2.896)(50,3.996)};
\legend{{\algobffmm},{\algobffaa},{\algobffgreedy},$DCS$}
\end{axis}
\end{tikzpicture}}
\end{tabular}}
\caption{Synthetic dataset ($p_A =0.5)$: execution time of the different algorithms for the {\problemma} problem for varying number of (a) nodes,  and (b) snapshots.}
\vspace{-0.2cm}
\label{fig:synth_times}
\end{figure}

\input{plots} 
	
\subsection{{\large {\problemtbff}} evaluation}
In this set of experiments, we
evaluate the performance of the iterative and incremental {\algotbff} algorithms.

\spara{Comparison of the algorithms in terms of solution quality:} Similar to before, we plant a dense random graph $A$ in $k$ snapshots.
We then run the {\algotbff} algorithms with the same value of $k$.
In Fig. \ref{fig:syntho2_1}, we report the $F$ measure for the different values of $k$ expressed as a percentage of the total number of snapshots.
For the iterative {\algotbff} algorithm, the {\it at-least-k}  initialization ({\itk}) outperforms the other two, and it successfully locates $A$ for all four density definitions, when $A$ appears in a sufficient number of snapshots.
Non-surprisingly, all initializations work equally well for average aggregation over time
(i.e., {\problemtbffam} and {\problemtbffaa}).
For the incremental {\algotbff} algorithm,  {\it density} ({\ind}) slightly outperforms
{\it overlap} ({\ino}).
Overall, the incremental algorithms achieve highest $F$, when compared with the iterative ones.

We also conduct a second experiment in which we plant a dense random graph $A$ with edge probability $p_A$ = 0.5 in all snapshots and a dense random graph $B$ with edge probability $p_B$ = 0.9 in $k$  snapshots.
In Fig. \ref{fig:syntho2}, we report the $F$ measure assuming that $B$ is the correct output for the
{\problemtbff} problem for different values of $k$ expressed as a percentage of the total number of snapshots.
Again, by comparing the different initializations for the iterative {\algotbff} algorithm, we observe that among the iterative algorithms, {\itk}  successfully locates $B$ for all four density definitions, when $B$ appears in a sufficient number of snapshots.
As in the previous experiment, all initializations work equally well for average aggregation over time.
The incremental algorithms 
outperform the iterative ones with {\ind} being the champion, since they achieve higher $F$ measure values even when $B$ appears in a few snapshots.

We also apply  the {\algotbff} algorithms on all real datasets for various values of $k$.
In Figs. \ref{fig:dbten} -- \ref{fig:ort}, we report the value of the aggregate density for {\dbten}, {\orf}, and {\ort} for different values of $k$, again expressed as a percentage of the total number of snapshots of the input graph history.
Overall, we observed that, in contradistinction to the experiments with real datasets, the {\it contiguous} initialization ({\itc}) of the iterative {\problemtbffaa} algorithm emerges as the best algorithm in many cases, slightly outperforming  {\ind}.
This is indicative of \emph{temporal locality}
of dense subgraphs in these datasets,
i.e., in these datasets dense subgraphs are usually alive in a few contiguous snapshots.
This is especially evident in datasets from collaboration networks such as the {\db} datasets.
We also notice that the incremental algorithms find solutions with density very close to that of the iterative algorithms.
Finally, we also observe that as $k$ increases the aggregate density of the solutions decrease.
This again is explained by the fact that often dense subgraphs
are only ``alive" in a few snapshots.

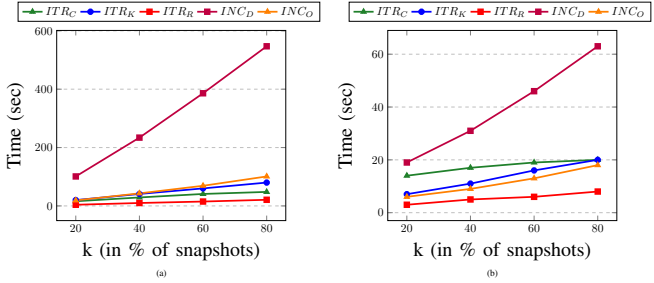
\begin{figure}
\centering
\resizebox{1.0\columnwidth}{!}{
\begin{tabular}{cc}
\subfloat[]
{\begin{tikzpicture}
\begin{axis}[
	xtick = {20, 40, 60, 80},
	xlabel={k (in \% of snapshots)},
	ylabel={Time (sec)},
	y label style={font=\LARGE},
    x label style={font=\LARGE},
    scaled y ticks = false,
	legend style={legend columns=5, anchor=north,at={(0.49,1.15)}},
	ymajorgrids=true,
	grid style=dashed,
]
	\addplot[color=darkgreen, mark=triangle*,line width=1.5pt,mark size=2pt]coordinates{(20,16)(40,29)(60,41)(80,48)};
	\addplot[color=blue, mark=*,line width=1.5pt,mark size=2pt]coordinates{(20,20)(40,41)(60,60)(80,80)};
	\addplot[color=red, mark=square*,line width=1.5pt,mark size=2pt]coordinates{(20,4)(40,10)(60,15)(80,21)};
	\addplot[color=purple, mark=square*,line width=1.5pt,mark size=2pt]coordinates{(20,101)(40,234)(60,386)(80,547)};
\addplot[color=orange, mark=triangle*,line width=1.5pt,mark size=2pt]coordinates{(20,19)(40,43)(60,69)(80,101)};
\legend{$ITR_{C}$,${ITR_K}$,${ITR_R}$,${INC_D}$,${INC_O}$}
\end{axis}
\end{tikzpicture}}
&
\subfloat[]
{\begin{tikzpicture}
\begin{axis}[
	xtick = {20, 40, 60, 80},
	xlabel={k (in \% of snapshots)},
	ylabel={Time (sec)},
	y label style={font=\LARGE},
    x label style={font=\LARGE},
	legend style={legend columns=5, anchor=north,at={(0.49,1.15)}},
	ymajorgrids=true,
	grid style=dashed,
]
	\addplot[color=darkgreen, mark=triangle*,line width=1.5pt,mark size=2pt]coordinates{(20,14)(40,17)(60,19)(80,20)};
	\addplot[color=blue, mark=*,line width=1.5pt,mark size=2pt]coordinates{(20,7)(40,11)(60,16)(80,20)};
	\addplot[color=red, mark=square*,line width=1.5pt,mark size=2pt]coordinates{(20,3)(40,5)(60,6)(80,8)};
	\addplot[color=purple, mark=square*,line width=1.5pt,mark size=2pt]coordinates{(20,19)(40,31)(60,46)(80,63)};
	\addplot[color=orange, mark=triangle*,line width=1.5pt,mark size=2pt]coordinates{(20,6)(40,9)(60,13)(80,18)};
\legend{$ITR_{C}$,${ITR_K}$,${ITR_R}$,${INC_D}$,${INC_O}$}
\end{axis}
\end{tikzpicture}}
\end{tabular}}
\caption{Execution time of the different algorithms for the {\problemtbffmm} problem in (a) Synthetic ($p_A = 0.5$, $p_B = 0.9$), and (b) {\ort} datasets.}
\vspace{-0.3cm}
\label{fig:times}
\end{figure}

\spara{Convergence and running time:} In terms of convergence, iterative
{\algotbff} requires 2-6 iterations to converge in all datasets.
In Fig. \ref{fig:times} we report the execution time of {\problemtbff} algorithms for the {\problemmm} problem in synthetic ($p_A = 0.5$, $p_B = 0.9$), and (b) {\ort} datasets.
As we observed, iterative and incremental {\ino} algorithms scale well with $k$.
Comparing incremental algorithms, {\ino} is up to 6x and 3.5x faster than {\ind} in synthetic and {\ort} datasets respectively due to the quadratic complexity of the latter.
Additional experiments including ones with  synthetic datasets with larger graphs and more intervals are depicted in Fig. \ref{fig:synth_times_o2}(a) and Fig. \ref{fig:synth_times_o2}(b) respectively.
In particular, Fig. \ref{fig:synth_times_o2}(a) shows the execution time of the different algorithms for the {\problemtbffmm} problem for varying number of nodes, with $\tau = 10$ and $k = 6$
whereas Fig. \ref{fig:synth_times_o2}(b) shows the execution time  for varying number of snapshots with $k = \frac{1}{6}~\tau$.

\spara{Summary:}  In conclusion, all algorithms successfully discovered the planted dense subgraphs
that lasted a sufficient percentage  (much less than half) of the snapshots with the incremental ones being more sensitive.
Among the {\algotbff} algorithms,  
incremental algorithms outperform the iterative ones in most cases.
Among the incremental algorithms, {\ind} is slightly better than {\ino}. However,
given the slow running time of {\ind}, {\ino} is a more preferable choice.
Finally, in datasets consisting of dense subgraphs with temporal locality, {\itc} is a good choice for detecting such graphs.

\begin{figure}[h!]
\centering
\resizebox{1.0\columnwidth}{!}{
\begin{tabular}{cc}
\subfloat[]
{\begin{tikzpicture}
\begin{axis}[
	xtick = {8, 16, 24, 32, 40},
	xlabel={\# nodes in thousands},
	ylabel={Time (sec)},
	y label style={font=\LARGE},
    x label style={font=\LARGE},
    ymode=log,
    log basis y={2},
	legend style={legend columns=5, anchor=north,at={(0.49,1.15)}},
	ymajorgrids=true,
	grid style=dashed,
]
	\addplot[color=darkgreen, mark=triangle*,line width=1.5pt,mark size=2pt]coordinates{(8,13.156)(16,46.809)(24,130.671)(32,221.549)(40,443.910)};
	\addplot[color=blue, mark=*,line width=1.5pt,mark size=2pt]coordinates{(8,30.659)(16,137.984)(24,391.808)(32,793.797)(40,1489.815)};
	\addplot[color=red, mark=square*,line width=1.5pt,mark size=2pt]coordinates{(8,16.550)(16,47.828)(24,161.435)(32,264.169)(40,342.921)};
\addplot[color=purple, mark=square*, line width=1.5pt,mark size =2pt]coordinates{(8,94.786)(16,380.322)(24,1175.826)(32,2671.726)(40,4567.362)};
\addplot[color=orange, mark=x,line width=1.5pt,mark size=2pt]coordinates{(8,46.33)(16,173.963)(24,596.999)(32,1127.930)(40,1987.734)};
\legend{$ITR_{C}$,${ITR_K}$,${ITR_R}$,${INC_D}$,${INC_O}$}
\end{axis}
\end{tikzpicture}}
&
\subfloat[]
{\begin{tikzpicture}
\begin{axis}[
	xtick = {10, 20, 30, 40, 50},
	xlabel={\# snapshots},
	ylabel={Time (sec)},
	y label style={font=\LARGE},
    x label style={font=\LARGE},
    ymode=log,
    log basis y={2},
	legend style={legend columns=5, anchor=north,at={(0.49,1.15)}},
	ymajorgrids=true,
	grid style=dashed,
]
	\addplot[color=darkgreen, mark=triangle*,line width=1.5pt,mark size=2pt]coordinates{(10,9.775)(20,31.814)(30,112.358)(40,210.370)(50,378.542)};
	\addplot[color=blue, mark=*,line width=1.5pt,mark size=2pt]coordinates{(10,14.201)(20,60.867)(30,112.358)(40,210.370)(50,372.628)};
	\addplot[color=red, mark=square*,line width=1.5pt,mark size=2pt]coordinates{(10,4.27)(20,11.586)(30,18.338)(40,28.004)(50,35.010)};
\addplot[color=purple, mark=square*, line width=1.5pt,mark size =2pt]coordinates{(10,89.321)(20,899.116)(30,2900)(40,7014)(50,13923)};
\addplot[color=orange, mark=x,line width=1.5pt,mark size=2pt]coordinates{(10,19.109)(20,86.301)(30,175.015)(40,336.885)(50,573.490)};
4.270
\legend{$ITR_{C}$,${ITR_K}$,${ITR_R}$,${INC_D}$,${INC_O}$}
\end{axis}
\end{tikzpicture}}
\end{tabular}}
\caption{Synthetic dataset ($p_A =0.5)$: execution time (log scale) of the different algorithms for the {\problemtbffmm} problem for varying number of (a) nodes, and (b) snapshots.}
\label{fig:synth_times_o2}
\end{figure}
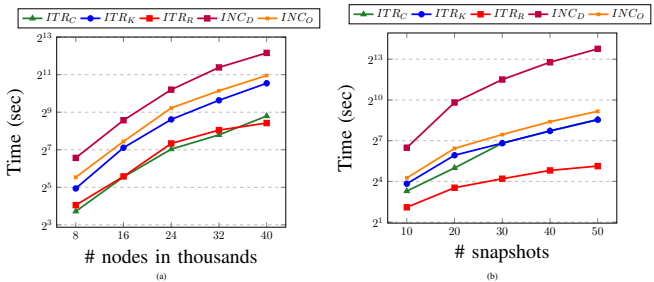

\begin{table}
\caption{The {\problembff} solutions for {\dbten} (in parenthesis dense author subgroups).\vspace{-0.2cm}}
\begin{center}
\resizebox{1.0\columnwidth}{!}{
\begin{tabular}{p{1.0\columnwidth}}
\hline
\multicolumn{1}{c}{\textit{\problemmm}}\\
\hline
\small {(Wei Fan, Philip S. Yu, Jiawei Han, Charu C. Aggarwal), (Lu Qin,
Jeffrey Xu Yu, Xuemin Lin), (Guoliang Li, Jianhua Feng),
(Craig Macdonald, Iadh Ounis)}\\
\hline
\hline
\multicolumn{1}{c}{\textit{\problemma}}\\
\hline
\small {(Wei Fan, Jing Gao, Philip S. Yu, Jiawei Han, Charu C. Aggarwal),
(Jeffrey Xu Yu, Xuemin Lin, Ying Zhang)}\\
\hline
\hline
\multicolumn{1}{c}{\textit{\problemam}}\\
\hline
\multicolumn{1}{c}{\small {(Wei Fan, Jing Gao, Philip S. Yu, Jiawei Han)}}\\
\hline
\hline
\multicolumn{1}{c}{\textit{\problemaa}}\\
\hline
\small {(Wei Fan, Jing Gao, Philip S. Yu, Jiawei Han, Charu C. Aggarwal, Mohammad M. Masud, Latifur Khan, Bhavani M. Thuraisingham)}\\
\hline
\hline
\end{tabular}}
\end{center}
\label{table:authors}
\vspace{-0.3cm}
\end{table}

\subsection{Case studies}
In this section, we report indicative results
we obtained using the {\dbten} and the
{\twitter}  datasets. These results identify  lasting dense author collaborations  and hashtag co-occurrences respectively.

\spara{Lasting dense co-authorships in {\dbten}:}
In Table~\ref{table:authors}, we report the set of nodes
output as solutions to the different {\problembff} problem variants,
on the {\dbten} dataset.

First, observe that three authors ``Wei Fan'', ``Philip S. Yu'', and ``Jiawei Han''
are part of \emph{all} four solutions.
These three authors have co-authored only two papers together in our dataset, but pairs of them have collaborated very frequently over the last decade.
The solutions for {\problemam} and {\problemaa} contain additional collaborators of these authors.
For  {\problemaa} we obtain a solution of 8 authors. Although, this group has no paper in which they are all co-authors, 
subsets of the authors have collaborated
with each other in many snapshots, resulting in high value of $f_\saa$.
The solutions for {\problemmm} and {\problemma} contain the aforementioned three authors and some of their collaborators, but also some new names.  
These are authors that have scarce or no collaborations with the former group.
Thus, in this case, the solutions consist of more than one dense subgroups of authors (grouped in parentheses), that are densely connected within themselves, but sparsely or not connected with others, while this is not the case for {\problemam} and {\problemaa}.

In Table \ref{table:authors_o2}, 
we report results for {\problemtbffmm}, {\problemtbffma}, {\problemtbffam} and {\problemtbffaa} on the same dataset. 
These authors are the most dense collaborators for $k$ = 2, 4, 6, and 8 (recall there are 10 years in the dataset). We also report the corresponding years of their dense collaborations. 
Many new groups of authors appear.
For example, we have new groups of collaborators from Tsinghua University, CMU and RPI among others.
The authors appeared in the solutions of  {\problembff}  
also appear here for large values of $k$.

\begin{table*}[ht!]
\caption{The authors output as solutions to the {\problemtbff} problem
on  {\dbten}.}
\begin{center}
\resizebox{0.9\linewidth}{!}{
\begin{tabular}{|l|l|l|l|l|l|l|l|l|l|}
\hline
\textbf{\small k = 2} & \multicolumn{9}{p{1.8\columnwidth}|}{\textit{\problemmm}, \textit{\problemma}, \textit{\problemam}, \textit{\problemaa}}\\
\hline
 &
\multicolumn{9}{p{1.8\columnwidth}|}{\small Christos Faloutsos, Leman Akoglu, Lei Li, Keith Henderson, Hanghang Tong, Tina Eliassi-Rad} \\
\textsl{\small Years:} & \multicolumn{9}{l|}{\textsl{\small 2010 - 2011}} \\
\hline \hline
\textbf{\small k = 4}& 
\multicolumn{4}{p{0.8\columnwidth}|}{\textit{\problemmm}, \textit{\problemma}, \textit{\problemam}} & \multicolumn{5}{p{\columnwidth}|}{\textit{\problemaa}}\\
\hline
& \multicolumn{4}{p{0.8\columnwidth}|}{\small Mo Liu, Chetan Gupta, Song Wang, Ismail Ari, Elke A. Rundensteiner} &
\multicolumn{5}{p{\columnwidth}|}{\small Yong Yu, Dingyi Han, Zhong Su, Lichun Yang, Shengliang Xu, Shenghua Bao)} \\
\textsl{\small Years:} & 
\multicolumn{4}{p{0.8\columnwidth}|}{\textsl{\small 2010 - 2013}} & 
\multicolumn{5}{p{\columnwidth}|}{\textsl{2007, 2009 - 2011}}\\
\hline\hline
\textbf{\small k = 6}& 
\multicolumn{4}{p{0.8\columnwidth}|}{\textit{\problemmm}, \textit{\problemma}, \textit{\problemam}} & \multicolumn{5}{p{\columnwidth}|}{\textit{\problemaa}}\\
\hline
& \multicolumn{4}{p{0.8\columnwidth}|}{\small Liyun Ru, Min Zhang, Yiqun Liu, Shaoping Ma}&
\multicolumn{5}{p{\columnwidth}|}{\small Min Zhang, Liyun Ru Bhavani, Yiqun Liu, Shaoping Ma), Latifur Khan, M. Thuraisingham, Mohammad M. Masud, Wei Fan, Jing Gao, Philip S. Yu, Jiawei Han}\\
\textsl{\small Years:} & 
\multicolumn{4}{p{0.8\columnwidth}|}{\textsl{2007 - 2012}} & 
\multicolumn{5}{p{\columnwidth}|}{\textsl{2007 - 2012}}\\
\hline\hline
\textbf{\small k = 8}& 
\multicolumn{2}{p{0.4\columnwidth}|}{\textit{\problemmm}} &
\multicolumn{2}{p{0.4\columnwidth}|}{\textit{\problemma}} &
\multicolumn{2}{p{0.4\columnwidth}|} {\textit{\problemam}} &
\multicolumn{3}{p{0.6\columnwidth}|}{\textit{\problemaa}}\\
\hline
& \multicolumn{2}{p{0.4\columnwidth}|}{\small Min Zhang, Yiqun Liu, Shaoping Ma} &
\multicolumn{2}{p{0.4\columnwidth}|} {\small Wei Fan, Jing Gao, Philip S. Yu, Jiawei Han, Charu C. Aggarwal} &
\multicolumn{2}{p{0.4\columnwidth}|}{\small Liyun Ru, Min Zhang, Yiqun Liu, Shaoping Ma} &
\multicolumn{3}{p{0.6\columnwidth}|}{\small Latifur Khan, Bhavani M. Thuraisingham, Mohammad M. Masud, Wei Fan, Jing Gao, Philip S. Yu, Jiawei Han, Charu C. Aggarwal} \\
\textsl{\small Years:} &
\multicolumn{2}{p{0.4\columnwidth}|}{\textsl{\small 2007 - 2014}}  &
\multicolumn{2}{p{0.4\columnwidth}|}{\textsl{\small 2007 - 2008, 2010 - 2015}}  &
\multicolumn{2}{p{0.4\columnwidth}|}{\textsl{\small 2007 - 2014}} &
\multicolumn{3}{p{0.6\columnwidth}|}{\textsl{\small 2007 - 2012, 2014 - 2015}}\\
\hline
\end{tabular}}
\end{center}
\label{table:authors_o2}
\vspace{-0.4cm}
\end{table*}

We also studied experimentally the {\problemqbff} problem. In Table~\ref{table:q-authors}, we show indicative results for three of the authors of this paper as seed nodes. For \textit{E. Pitoura}, we retrieve a group of ex-graduate students with whom she had a lasting and prolific collaboration; for \textit{E. Terzi}  close collaborators from BU University, and for \textit{P. Tsaparas}, a group of collaborators from his time at Microsoft Research. Note that in the last case, the selected set consists of researchers with whom \textit{P. Tsaparas} has co-authored several papers in the period recorded in our dataset, but these
authors are also collaborating amongst themselves.
Finally, we use one of the authors appearing in the dense subgraphs of the {\problemtbff}, namely \textit{C. Faloutsos} as seed node. In this case, we obtain a dense subgraph similar to the one we have reported in Table~\ref{table:authors_o2}.
 Finally, we consider a query with two authors: \textit{C. Faloutsos} and his student \textit{D. Koutra}. Adding \textit{D. Koutra} to the query set changes the consistency of the result, focusing more on authors that are collaborators of both query nodes. 

\spara{Lasting dense hashtag appearances in {\twitter}:}
In Table \ref{table:twitter}, we report results of the {\problemtbff} problem on the {\twitter} dataset.
Note that the results  of the {\problembff} problem on this dataset (as shown in Table \ref{table:size-dens}) are very small graphs, since very few hashtags appear together in all 15 days of the dataset.
As seen in Table \ref{table:twitter}, we were able to discover interesting dense subgraphs of hashtags appearing in $k$ = 3, 6, and 9 of these days. These hashtags correspond to actual events (including f1 races and wikileaks) that were trending during that period.

For each solution, we also report the selected snapshot dates. As expected there is time-contiguity in the selected dates, but our approach also captures the interest fluctuation over time. For example, for the wikileaks topic that is captured in the dense hashtag set \{``wikileaks'', ``snowden'', ``nsa'', ``prism''\}, the best snapshots are collections of contiguous intervals, rather than a single contiguous interval.

Note also, that for large values of $k$, we do not get interesting results which is a fact consistent with the
ephemeral nature of Twitter, where hashtags are short-lived. This is especially true for  $f_\smm$ and $f_\sma$ that impose strict density constraints 
and as a result the solutions consist of disconnected edges.

\begin{table}
\caption{An example of authors output as solutions to the {\problemqbff} problem on  {\dbten}.\vspace{-0.2cm}}
\begin{center}
\begin{tabular}{p{0.9\columnwidth}}
\hline
\multicolumn{1}{c}{\textit{\problemqbff}}\\
\hline
\textbf{E. Pitoura}: G. Koloniari, M. Drosou, K. Stefanidis\\
\textbf{E.Terzi}: V. Ishakian, D. Erdos, A. Bestavros\\
\textbf{P. Tsaparas}: A. Fuxman, A. Kannan, R. Agrawal\\
\textbf{C. Faloutsos, D. Koutra}: Chris H. Q. Ding, L. Akoglu, H. Huang, Lei Li, Tao Li, H. Tong
\\
\hline
\hline
\end{tabular}
\end{center}
\vspace{-0.4cm}
\label{table:q-authors}
\end{table}

\begin{table*}
\caption{The hashtags and the chosen snapshot dates output as solutions to the {\problemtbff} problem on {\twitter}.\vspace{-0.2cm}}
\begin{center}
\resizebox{0.9\linewidth}{!}{
\begin{tabular}{|l|p{0.42\columnwidth}|p{0.7\columnwidth}|p{0.8\columnwidth}|}
\hline
\textbf{\small k = 3}& {\textit{\problemmm}, \textit{\problemma}} & {\textit{\problemam}} & {\textit{\problemaa}}\\
\hline
 &
{\small kimi, abudhabigp, f1, allowin} &
{\small ozpol, nz, mexico, malaysia, signapore, vietnam, chile, peru, tpp, japan, canada} &
{\small abudhabigp, fp1, abudhabi, guti, f1, pushpush, skyf1, hulk, allowin, bottas, kimi, fp3, fp2} \\
\textsl{\small Dates:} & \textsl{\small Oct 31-Nov 2}  & \textsl{\small Oct 27-28, Nov 7} & \textsl{\small Oct 31-Nov 2} \\
\hline\hline
\textbf{k = 6}& {\textit{\problemmm}, \textit{\problemma}} &{\textit{\problemam}} & {\textit{\problemaa}}\\
\hline
&
{\small abudhabigp, f1, skyf1} &
{\small wikileaks, snowden, nsa, prism} &
{\small abudhabigp, fp1, abudhabi, guti, f1, pushpush, skyf1, hulk, allowin, bottas, kimi, fp3, fp2} \\
\textsl{\small Dates:} & \textsl{\small Oct 28-Nov 2} & \textsl{\small Oct 27-28, Nov 3,5,7} & \textsl{\small Oct 28, Oct 30-Nov 1, Nov 9}\\
\hline\hline
\textbf{\small k = 9}& {\textit{\problemmm}, \textit{\problemma}} & {\textit{\problemam}} & {\textit{\problemaa}}\\
\hline
 &
{\small (Too many tags to report)} &
{\small wikileaks, snowden, nsa, prism} &
{\small assange, wikileaks, snowden, nsa, prism}\\
\textsl{\small Dates:} &    & \textsl{\small Oct 27-31, Nov 3,5-7} & \textsl{\small Oct 27-29,31, Nov 3,5-7,10}\\
\hline\hline
\textbf{\small k = 12}& {\textit{\problemmm}, \textit{\problemma}} & {\textit{\problemam}} & {\textit{\problemaa}}\\
\hline
&
{\small  (Too many tags to report)} &
{\small wikileaks, snowden, nsa } &
{\small assange, wikileaks, snowden, nsa, prism} \\
\textsl{\small Dates:} &  & \textsl{\small Oct 27-Nov 1, Nov 3-7,10} & \textsl{\small Oct 27-31, Nov 2-7, 10}\\
\hline
\end{tabular}}
\end{center}
\vspace{-0.4cm}
\label{table:twitter}
\end{table*}
\section{Related Work}
\label{sec:rw}

To the best of our knowledge, we are the first to systematically study all the variants of the {\problembff}, and {\problemtbff} problems.

The research most related to ours is the recent work of
Jethava and
Beerenwinkel~\cite{jethava15finding} and Rozenshtein {\etal}~\cite{rozenshtein14discovering}.
To the best of our understanding, the authors of~\cite{jethava15finding} introduce one of the four variants of the
 {\problembff} problem we studied here, namely, {\problemma}.
 In their paper, the authors
 conjecture that the problem is NP-hard and they propose a heuristic algorithm.
 Our work performs a rigorous and systematic study of the general {\problembff} problem for multiple variants of the aggregate density function.
Additionally, we introduce and study the {\problemtbff} problem, which is not studied in~\cite{jethava15finding}.
The authors of ~\cite{rozenshtein14discovering} study a problem that can be considered a special case of the {\problemtbff} problem.
In particular, their goal is to identify a subset of nodes that are
dense in the graph consisting of the union of edges appearing
in the selected snapshots, which is a weak definition of aggregate
density. Furthermore, they focus on finding collections of
contiguous intervals, rather than arbitrary snapshots.
They propose an algorithm similar to the iterative algorithm we consider, which we have shown  to be outperformed by the incremental algorithms.
There is a huge literature on extracting ``dense'' subgraphs from a single
graph snapshot.
Most formulations for finding subgraphs that define near-cliques are often NP-hard and often hard to approximate due to their
connection to the maximum-clique problem~\cite{DBLP:conf/nips/Alvarez-HamelinDBV05,DBLP:journals/eor/BourjollyLP02,makino04new,DBLP:journals/jco/McCloskyH12,tsourakakis13denser}.
As a result, the problem of finding the subgraph with the maximum average or minimum degree
has become particularly popular, due to its computational tractability.
Specifically, the problem of finding a subgraph with the maximum average degree can be solved optimally in polynomial
time~\cite{DBLP:conf/approx/Charikar00,goldberg1984finding,DBLP:conf/icalp/KhullerS09}, and there exists a practical greedy algorithm
that gives a $2$-approximation guarantee in time linear to the number of edges and nodes
of the input graph~\cite{DBLP:conf/approx/Charikar00}.
The problem of identifying a subgraph with the maximum minimum degree, can be solved optimally in polynomial time~\cite{DBLP:conf/kdd/SozioG10}, using again
the greedy algorithm proposed by Charikar~\cite{DBLP:conf/approx/Charikar00}.
In our work, we use the average and minimum degree to quantify the density of the subgraph in a single graph snapshot, and we extend these definitions to sets of snapshots.
The algorithmic techniques we use for the {\problembff} problem are inspired by the techniques
proposed by Charikar~\cite{DBLP:conf/approx/Charikar00}, and by Sozio and Gionis~\cite{DBLP:conf/kdd/SozioG10}; however, adapting them to handle multiple snapshots is non-trivial.

Existing work also studies the problem of identifying a dense subgraph on time-evolving graphs~\cite{DBLP:conf/www/EpastoLS15,bahmani12densest,bhattacharya15spacetime}; these are
graphs where new nodes and edges may appear over time and existing ones may disappear.
The goal in this line of work is to devise a \emph{streaming algorithm} that at any point in time it reports the densest subgraph for the current version of the graph.
In our work, we are not interested in the dynamic version of the
problem and thus the algorithmic challenges that our problem raises are orthogonal to those
faced by the work on streaming algorithms.
Other recent work \cite{DBLP:conf/icde/MaHWLH17} focuses on
detecting dense subgraphs in
a special class of temporal weighted networks with fixed nodes and edges, where edge weights change over time and may take both positive and negative values.
This is a different problem, since we consider graphs with changing edge sets. Furthermore, density in the presence of edges with negative weights is  different than density when edges have only positive weights.

Finally, another line of research focuses
on processing queries  e.g., reachability, path distance, graph matching, etc. over multiple graph snapshots \cite{DBLP:conf/edbt/SemertzidisPL15,DBLP:conf/www/MoffittS16,DBLP:conf/icde/SemertzidisP16, DBLP:conf/icde/KhuranaD13,DBLP:journals/pvldb/RenLKZC11}.
The main goal of this work is to devise effective storage, indexing and retrieving  techniques so that queries over such sequences of graphs are answered efficiently.
In this paper, we propose a novel problem that of finding dense subgraphs.

\section{Summary}
\label{sec:conclusions}

In this paper, we introduced and systematically studied the problem of identifying dense subgraphs
in a collection of graph snapshots defining a graph history.  We showed that for many definitions of
aggregate density functions the problem of identifying a subset of nodes that are densely-connected
in \emph{all} snapshots (i.e., the {\problembff} problem) can be solved in linear time.  We also demonstrated that other versions of the {\problembff} problem (i.e., {\problemma} and {\problemam}) cannot be solved with the same algorithm.
To identify dense subgraphs that occur in $k$, yet not all, the snapshots of a graph history
we also defined the {\problemtbff} problem.  For all variants of this problem we showed that they are NP-hard and we devised an iterative and an incremental algorithm for solving them.
Our extensive experimental evaluation with datasets from diverse domains 
demonstrated the effectiveness and the efficiency of our algorithms.

\bibliographystyle{IEEEtran}
\bibliography{_paper}
\newpage
\section{Appendix}
\label{appendix:examples}

In this section we  present counter-examples that demonstrate that the {\algobffmm} and {\algobffaa} when applied to the {\problemma} and {\problemaa} yield a solution that is a poor approximation of the optimal solution. For the following, we use $n = |V|$ to denote the number of nodes in the different snapshots, and $\tau = |\calG|$ to denote the number of snapshots.
\vspace{-0.05in}
\subsection{Proof of Proposition 4}
\begin{proof}
In order to prove our claim we need to construct an instance of the {\problemam} problem where the {\algobffmm} algorithm produces a solution with approximation ratio $O\left(\frac{1}{n}\right)$. We construct the graph history $\calG = \{G_1,...,G_\tau\}$ as follows. The first $\tau -1$ snapshots consist of a full clique with $n-1$ nodes, plus an additional node $v$ that is connected to a single node $u$ from the clique. The last snapshot $G_\tau$ consists of just the edge $(v,u)$.

In the first $n-2$ iterations of the {\algobffmm} algorithm, the node with the minimum minimum degree is one of the nodes in the clique (other than the node $u$). Thus the nodes in the clique will be iteratively removed, until we are left with the edge $(u,v)$. Since node $v$ is present in all intermediate subsets $S_i$, the minimum degree in all snapshots $G_t$ is 1. Therefore, the solution $S$ of the {\algobffmm} algorithm has $f_{am}(S) = 1$. On the other hand clearly the optimal solution $S^*$ consists of the nodes in the clique, where we have minimum degree $n-2$, except of the last instance where the minimum degree is zero. Therefore, $f(S^*) = (n-2)\frac{\tau-1}{\tau}$ which proves our claim.
\end{proof}

\subsection{Proof of Proposition 5}
\begin{proof}
In order to prove our claim we need to construct an instance of the {\problemam} problem where the {\algobffaa} algorithm produces a solution with approximation ratio $O\left(\frac{1}{n}\right)$. We construct the graph history $\calG = \{G_1,...,G_\tau\}$, where $\tau$ is even, as follows. Each snapshot $G_t$ contains $n = 2b+3$ nodes. The $2b$ of these nodes form a complete $b\times b$ bipartite graph. Let $u$, $v$, and $s$ denote the additional three nodes. Node $s$ is connected to all nodes in the graph, in all snapshots, except for the last snapshot where $s$ is connected only to $u$ and $v$. Nodes $u$ and $v$ are connected to each other in all snapshots, and node $u$ is connected to all $2b$ nodes of the bipartite graph in the first $\tau/2$ snapshots, while node $v$ is connected to all $2b$ nodes of the bipartite graph in the last $\tau/2$ snapshots.
Throughout assume that $\tau\geq 2$.
Note that the optimal set $S^*$ for this history graph consists of the $2b$ nodes in the bipartite graph, with $f_\sam(S^*,\calG) = b = \Theta(n)$.

The score $\textit{score}_a$ for every node $w$ of the $2b$ nodes in the bipartite graph is $\textit{score}_a(w,\calG) = b+1+\frac{\tau-1}{\tau}$. For the nodes $u$ and $v$, we have $\textit{score}_a(u,\calG) = \textit{score}_a(v,\calG) = \frac{2b\tau/2 + 2\tau}{\tau} = b+2$. Node $s$ has score $\textit{score}_a(s,\calG) = 2b\frac{\tau-1}{\tau} + 2$.

Therefore, in the first iteration, the algorithm will remove one of the nodes of the bipartite graph. Without loss of generality assume that it removes one of the nodes in the left partition. Now, for a node $w$ in the left partition, we still have that $\textit{score}_a(w,\calG[S_1]) = b+1+\frac{\tau-1}{\tau}$. For a node $w$ in the right partition we have that $\textit{score}_a(w,\calG[S_1]) = b + \frac{\tau-1}{\tau}$. For nodes $u$ and $v$ we have $\textit{score}_a(u,\calG[S_1]) = \textit{score}_a(v,\calG[S_1]) = \frac{(2b-1)\tau/2 + 2\tau}{\tau} = b+\frac 32$. For node $s$ we have that $\textit{score}_a(s,\calG[S_1]) = (2b-1)\frac{\tau-1}{\tau} + 2$.

Therefore, in the second iteration the algorithm will select to remove one of the nodes in the right partition.
Note that the resulting graph $\calG[S_2]$ is identical in structure with $\calG$, with $n = 2(b-1)+3$ nodes. Therefore, the same procedure will be repeated until all the nodes from the bipartite graph are removed, while nodes $u$ and $v$ will be kept in the set until the last iterations. As a result, the set $S$ returned by {\algobffaa} has $f_\sam(S,\calG) = 2$ (the degree of the nodes $u$ and $v$), yielding approximation ratio $O\left(\frac{1}{n}\right)$.
\end{proof}

\vspace{-0.05in}
\subsection{Proof of Proposition 6}
\begin{proof}
In order to prove our claim we need to construct an instance of the {\problemma} problem where the {\algobffmm} algorithm produces a solution with approximation ratio $O\left(\frac{1}{\sqrt{n}}\right)$. We construct the graph history $\calG = \{G_1,...,G_\tau\}$ as follows. We have $\tau=m$ snapshots that are all identical. They consist of two sets of nodes $A$ and $B$ of size $m$ and $m^2$ respectively. The nodes in $B$ form a cycle. The nodes in $A$ in graph snapshot $G_t$ form a clique with all nodes except for one node $v_t$, different for each snapshot. The optimal set $S^*$ consists of the nodes in $A$, that have average degree $\frac{(m-1)(m-2)}{m} = \Theta(m)$.

The {\algobffmm} starts with the set of all nodes. The average degree of any snapshot is $\frac{2m^2 + (m-1)(m-2)}{m^2+m} = \Theta(1)$, which is also the value of the $f_{ma}(V)$ function. In the first $m$ iterations of the algorithm, the nodes in $A$ have $\textit{score}_m(v,S_i) = 0$, so these are the ones to be removed first. Then the nodes in $B$ are removed. In all iterations the average degree in each snapshot remains $O(1)$. Therefore, the set $S$ returned by the {\algobffmm} has $f_{ma}(S) = \Theta(1)$, and the approximation ratio is $\Theta\left(\frac{1}{m}\right)$. Since $m = \sqrt{n}$, this proves our claim.
\end{proof}

\vspace{-0.05in}
\subsection{Proof of Proposition 7}
\begin{proof}
In order to prove our claim we need to construct an instance of the {\problemma} problem where the {\algobffaa} algorithm produces a solution with approximation ratio $O\left(\frac{1}{\sqrt{n}}\right)$. The construction of the proof is very similar to before. We construct the graph history $\calG = \{G_1,...,G_\tau\}$ as follows. We have $\tau=m$ snapshots that are all identical, except for the last snapshot $G_m$. The snapshots $G_1,...,G_{m-1}$ consist of two sets of nodes $A$ and $B$ that form two complete cliques of size $m$ and $m^2$ respectively. In the last snapshot the nodes in $B$ are all disconnected. The optimal set $S^*$ consists of the nodes in $A$, that have $f_{ma}(A) = \frac{m(m-1)}{m} = \Theta(m-1)$.

The {\algobffaa} starts with the set of all nodes. The value of $f_{ma}(V)$ is determined by the last snapshot $G_m$ that has average degree $\frac{m(m-1)}{m^2+m} = \Theta(1)$. The nodes in $A$ have average degree (over time) $\frac{m(m-1)}{m} = \Theta(m)$, while the nodes in $B$ have average degree $\frac{(m-1)(m^2-1)}{m} = \Theta(m^2)$. Therefore, the algorithm will iteratively remove all nodes in $A$. In each iteration the resulting set $S_i$ has $f_{ma}(S_i) = O(1)$. When all the nodes in $A$ are removed, we have that $f_{ma}(S_i) = 0$. Therefore, the approximation ratio for this instance is $\Theta(\frac{1}{m})$. Our claim follows from the fact that $n = m^2+m$.
\end{proof}

\end{document}